\documentclass[a4paper,10pt]{article}
\usepackage[utf8]{inputenc}
 \usepackage{amsmath,amsfonts,amssymb}
 \usepackage{amsthm}
 \usepackage{mathtools}
 \usepackage{amsopn}
 \usepackage{tabularx,lipsum,environ}
 \usepackage{paralist}
 \usepackage{microtype}
 \usepackage{authblk}
 \usepackage{hyperref}
 \usepackage{fancyhdr}
 \usepackage{rotating}
 \usepackage{overpic}
 \usepackage{ucs}
 \usepackage{enumerate}
 \usepackage{graphicx}
 \usepackage{booktabs}
 \usepackage{varwidth}
 \usepackage{subfig}

\title{Structurally Parameterized $d$-Scattered Set}
\author{Ioannis Katsikarelis}
\author{Michael Lampis}
\author{Vangelis Th. Paschos}
\affil{Université Paris-Dauphine, PSL Research University, CNRS, UMR 7243 \\ LAMSADE, 75016, Paris, France, \texttt{\{ioannis.katsikarelis|michail.lampis|paschos\}@lamsade.dauphine.fr}}

\theoremstyle{plain}
\newtheorem{theorem}{Theorem}
\newtheorem{lemma}[theorem]{Lemma}
\newtheorem{corollary}[theorem]{Corollary}

\newlength{\defbaselineskip}
\newcommand{\setlinespacing}[2]%
           {\setlength{\baselineskip}{#1 \defbaselineskip}}

\newlength{\btw}
\newlength{\stw}
\newsavebox\tmpbox

\newcommand{\tw}{\ensuremath\textrm{tw}}
\newcommand{\pw}{\ensuremath\textrm{pw}}
\newcommand{\td}{\ensuremath\textrm{td}}
\newcommand{\vc}{\ensuremath\textrm{vc}}
\newcommand{\fvs}{\ensuremath\textrm{fvs}}
\newcommand{\ms}{\ensuremath\textrm{ms}}
\newcommand{\KC}{$(k,r)$-\textsc{Center}}

\newcommand{\dS}{$d$-\textsc{Scattered Set}}

\date{}

\begin{document}

\maketitle

\begin{abstract} In \dS\ we are given an (edge-weighted) graph and are
asked to select at least $k$ vertices, so that the distance between any pair is at least $d$, thus generalizing \textsc{Independent Set}. We provide upper and lower bounds on the complexity of this problem with respect to
various standard graph parameters. In particular, we show the following:
\begin{itemize}
\item For any $d\ge 2$, an $O^*(d^{\tw})$-time algorithm, where $\tw$ is the treewidth of the input graph and a tight SETH-based lower bound matching this algorithm's performance. These generalize known results for \textsc{Independent Set}.
\item \dS\ is W[1]-hard parameterized by vertex cover (for edge-weighted graphs), or feedback vertex set (for unweighted graphs), even if $k$ is an additional parameter.
\item A single-exponential algorithm parameterized by vertex cover for unweighted graphs, complementing the above-mentioned hardness.
\item A $2^{O(\td^2)}$-time algorithm parameterized by tree-depth ($\td$), as well as a matching ETH-based lower bound, both for unweighted graphs.
\end{itemize}
We complement these mostly negative results by providing an \emph{FPT
approximation scheme} parameterized by treewidth. In particular, we give an
algorithm which, for any error parameter $\epsilon>0$, runs in time
$O^*((\tw/\epsilon)^{O(\tw)})$ and returns a $d/(1+\epsilon)$-scattered set of size $k$, if a $d$-scattered set of the same size exists.
\end{abstract}

\section{Introduction}\label{sec_intro} 

In this paper we study the \dS\ problem: given graph $G=(V,E)$ and a metric weight
function $w:E\mapsto\mathbb{N}^+$ that gives the length of each
edge, we are asked if there exists a set $K$ of at least $k$ \emph{selections}
from $V$, such that the distance between any pair $v,u\in K$ is at least
$d(v,u)\ge d$, where $d(v,u)$ denotes the shortest-path distance from $v$ to
$u$ under weight function $w$. If $w$ assigns weight 1 to all edges, the
variant is called \emph{unweighted}.

The problem can already be seen to be hard, as it generalizes \textsc{Independent Set} (for $d=2$),
even to approximate (under standard complexity assumptions), i.e.\ the optimal $k$ cannot be approximated to $n^{1-\epsilon}$ in polynomial time \cite{Hastad96}, while an alternative name is
\textsc{Distance-$d$ Independent Set} \cite{EtoILM16,MontealegreT16,EtoGM14}.
This hardness prompts the analysis of the problem when the input graph is of restricted structure, our aim being to provide a comprehensive account of the complexity of \dS\ through various upper and lower bound results.
Our viewpoint is parameterized: we consider the well-known
structural parameters treewidth {\bf$\tw$}, tree-depth {\bf$\td$}, vertex cover number
{\bf$\vc$} and feedback vertex set number {\bf$\fvs$}, that comprehensively express the intended restrictions on the input graph's structure (as they range in size and applicability), while we examine both the
edge-weighted and unweighted variants of the problem.

\paragraph{Our contribution:} First, in Section
\ref{sec_tw_lb} we present a lower bound of $(d-\epsilon)^{\tw}\cdot n^{O(1)}$ on the complexity of any algorithm solving \dS\ parameterized by $\tw$, based on the
Strong Exponential Time Hypothesis (SETH \cite{ImpagliazzoP01,ImpagliazzoPZ01}). This result can be seen as a non-trivial
extension of the bound of $(2-\epsilon)^{\tw}\cdot n^{O(1)}$ for \textsc{Independent Set} (\cite{LokshtanovMS11a}) for larger values of $d$, for which the construction is required to be much more compact in terms of encoded information per unit of treewidth. Next, in Section \ref{sec_tw_dp} we provide a dynamic
programming algorithm of running time $O^*(d^\tw)$, matching this lower bound, over a given tree
decomposition of width $\tw$. The algorithm actually solves the counting version of \dS, making use of standard techniques (dynamic programming on tree decompositions), with an application of the fast subset convolution technique of \cite{BjorklundHKK07} (or \emph{state changes} \cite{BodlaenderLRV10,RooijBR09}) to bring the running time down to match the size of the dynamic programming tables. 

Having thus identified the complexity of the problem with respect to $\tw$, we next focus on the more restrictive parameters $\vc$ and $\fvs$ and we show in Section \ref{sec_vc} that the edge-weighted \dS\ problem parameterized by $\vc+k$ is
W[1]-hard. If, on the other hand, all edge-weights are set to 1, then \dS\ (the unweighted variant) parameterized by $\fvs+k$ is W[1]-hard. Our reductions also imply lower bounds based on the Exponential Time Hypothesis (ETH
\cite{ImpagliazzoP01,ImpagliazzoPZ01}), yet we do not believe these to be tight, due to the quadratic increase in parameter size (as the construction's focus lies on the edges). One observation we can make is that there are few cases where we can expect to obtain an FPT algorithm without bounding the value of $d$.

We complement these results with a single-exponential algorithm for the
unweighted variant, of running time $O^*(3^{\vc})$ for the case of even $d$,
while for odd $d$ the running time is $O^*(4^{\vc})$. The algorithm is based on
defining a sub-problem based on a variant of \textsc{Set Packing} that we solve
via dynamic programming. The difference in running times, depending on the
parity of $d$, is due to the number of possible situations for a vertex with
respect to potential candidates for selection.

Further, for the unweighted variant we also show in
Section \ref{sec_td} the existence of an algorithm parameterized by $\td$ of running time $O^*(2^{O(\td^2)})$, as well as a matching ETH-based lower bound. The upper bound follows from known connections between the tree-depth
of a graph and its diameter, while the lower bound comes from a reduction from
\textsc{3-SAT}.

Finally, we turn again to $\tw$ in Section \ref{sec_tw_approx} and we present a
fixed-parameter-tractable approximation scheme (FPT-AS) on $d$ of running time
$O^*((\tw/\epsilon)^{O(\tw)})$, that finds a $d/(1+\epsilon)$-scattered set of size $k$, if a $d$-scattered set of the same size exists.  The algorithm is based on a
rounding technique introduced in \cite{Lampis14} and can be much faster than any exact algorithm for the problem (for large $d$, i.e.\ $d\ge
O(\log n)$), even for the unweighted case and more restricted parameters.
Figure \ref{fig:param_relations} illustrates the relationships between
considered parameters and summarizes our results.

\paragraph{Related work:}  Our work can be considered as a
continuation of the investigations in \cite{KatsikarelisLP17}, where the \KC\
problem is similarly studied with respect to several well-known structural
parameters and a number of fine-grained upper/lower bounds is presented, while
some of the techniques employed for our SETH lower bound are also present in
\cite{BorradaileL16}.

The SETH-based lower bound of $(2-\epsilon)^{\tw}\cdot n^{O(1)}$ on the running time of any algorithm for \textsc{Independent Set} parameterized by $\tw$ comes from \cite{LokshtanovMS11a}.
For \dS, Halld\'{o}rsson et al.\ \cite{HalldorssonKT00} showed a tight inapproximability ratio of $n^{1-\epsilon}$ for even $d$ and $n^{1/2-\epsilon}$ for odd $d$, while Eto et al.\
\cite{EtoILM16} showed that on $r$-regular graphs the problem is APX-hard for $r,d\ge3$, while also
providing polynomial-time $O(r^{d-1})$-approximations and a polynomial-time
approximation scheme (PTAS) for planar graphs.
For a class of graphs with at
most a polynomial (in $n$) number of minimal separators, \dS\ can be solved in
polynomial time for even $d$, 
while it remains NP-hard on chordal graphs (contained in the class) and
any odd $d\ge3$ \cite{MontealegreT16}.
It remains NP-hard even for planar bipartite graphs
of maximum degree 3, while a 1.875-approximation is available on cubic graphs
\cite{EtoILM17}. 
Several hardness results for planar and chordal (bipartite) graphs can be found in
\cite{EtoGM14}, while \cite{FominLRS11} shows the problem admits an
EPTAS on (apex)-minor-free graphs, based on the theory of bidimensionality.
Finally, on a related result, \cite{MarxP15} shows an
$n^{O(\sqrt{k})}$-time algorithm for planar graphs, making use of Voronoi diagrams and based on ideas previously used to obtain geometric QPTASs.

\begin{figure}[htbp]
\centerline{\includegraphics[width=110mm]{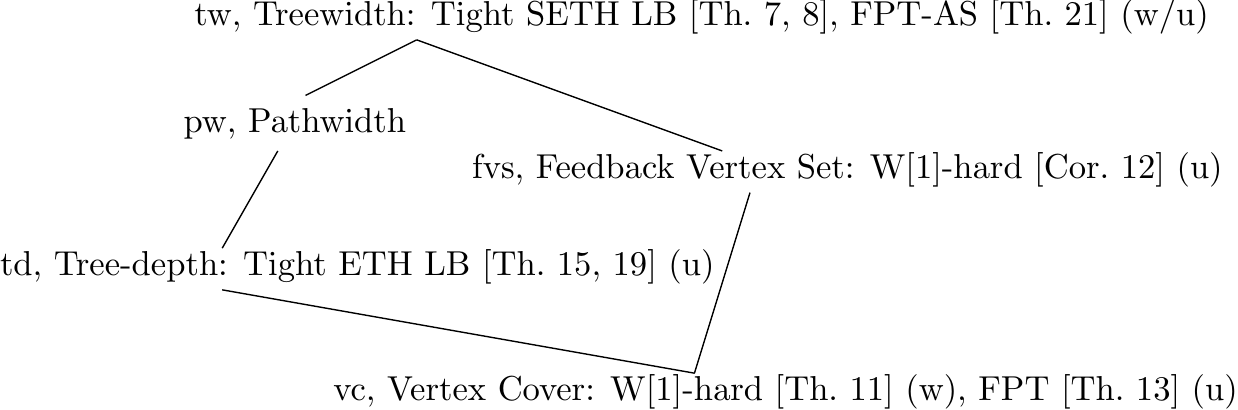}}
\caption{Relationships between parameters and an overview of our results (with theorem numbers, for weighted/unweighted variants). In the downwards direction (from $\tw$ to $\vc$) parameter size increases and algorithmic results are inherited, while hardness results are inherited upwards.}
\label{fig:param_relations}
\end{figure}

\section{Definitions and Preliminaries}\label{sec_prelim}
We use standard graph-theoretic notation. For a graph $G=(V,E)$, $n=|V|$
denotes the number of vertices, $m=|E|$ the number of edges, an edge $e\in E$
between $u,v\in V$ is denoted by $(u,v)$, and for a subset $X\subseteq V$,
$G[X]$ denotes the graph induced by $X$.
The functions $\lfloor x\rfloor$ and $\lceil x\rceil$, for $x\in\mathbb{R}$, denote the maximum integer that is not larger and the minimum integer that is not smaller than $x$, respectively.
Further, we assume the reader has some familiarity with standard definitions
from parameterized complexity theory (see
\cite{CyganFKLMPPS15,FlumG06,DowneyF13}).

For a parameterized problem with
parameter $k$, an FPT-AS is an algorithm which for any $\epsilon>0$ runs in
time $O^*(f(k,\frac{1}{\epsilon}))$ (i.e.\ FPT time when parameterized by
$k+\frac{1}{\epsilon}$) and produces a $(1+\epsilon)$-approximation (see \cite{Marx08}). We use $O^*(\cdot)$ to
imply omission of factors polynomial in $n$.
In this paper we present
approximation schemes with running times of the form $(\log
n/\epsilon)^{O(k)}$. These can be seen to imply an FPT running time by a
well-known win-win argument: 
\begin{lemma}\label{lem:fpt-logn}
If a parameterized problem with parameter $k$ admits, for some $\epsilon>0$, an
algorithm running in time $O^*((\log n /\epsilon)^{O(k)})$, then it also admits an
algorithm running in time $O^*((k/\epsilon)^{O(k)})$.
\end{lemma}
\begin{proof}
   We consider two cases: if $k\le \sqrt{\log n}$ then $(\log n/\epsilon)^{O(k)} =
  (1/\epsilon)^{O(k)} (\log n)^{O(\sqrt{\log n})} = O^*((1/\epsilon)^{O(k)})$. If
  on the other hand, $k> \sqrt{\log n}$ we have $\log n \le k^2$, so $O^*((\log n
  /\epsilon)^{O(k)}) = O^*((k/\epsilon)^{O(k)})$.
\end{proof}

\emph{Treewidth} and \emph{pathwidth} are standard notions in parameterized
complexity that measure how close a graph is to being a tree or path (\cite{Bodlaender00,Bodlaender06,Kloks94}).
A \emph{tree decomposition} of a graph $G=(V,E)$ is a pair $(\mathcal{X},T)$ with $T=(I,F)$ a tree and $\mathcal{X}=\{X_i|i\in I\}$ a family of subsets of $V$ (called \emph{bags}), one for each node of $T$, with the following properties:
 \begin{enumerate}[1)]
  \item $\bigcup_{i\in I}X_i=V$;
  \item for all edges $(v,w)\in E$, there exists an $i\in I$ with $v,w\in X_i$;
  \item for all $i,j,k\in I$, if $j$ is on the path from $i$ to $k$ in $T$, then $X_i\cap X_k\subseteq X_j$.
 \end{enumerate}
 The \emph{width} of a tree decomposition $((I,F),\{X_i|i\in I\})$ is $\max_{i\in I}|X_i|-1$. The \emph{treewidth} of a graph $G$ is the minimum width over all tree decompositions of $G$, denoted by $\textrm{tw}(G)$.

Moreover, for rooted $T$, let $G_i=(V_i,E_i)$ denote the \emph{terminal subgraph} defined by node $i\in I$, i.e.\ the induced subgraph of $G$ on all vertices in bag $i$ and its descendants in $T$. Also let $N_{i}(v)$ denote the neighborhood of vertex $v$ in $G_i$ and $d_i(u,v)$ denote the distance between vertices $u$ and $v$ in $G_i$, while $d(u,v)$ (absence of subscript) is the distance in $G$.

In addition, a tree decomposition can be converted to a \emph{nice} tree decomposition of the same width (in $O(\textrm{tw}^2\cdot n)$ time and with $O(\textrm{tw}\cdot n)$ nodes): the tree here is rooted and binary, while nodes can be of four types: 
\begin{enumerate}[a)]
 \item Leaf nodes $i$ are leaves of $T$ and have $|X_i|=1$;
 \item Introduce nodes $i$ have one child $j$ with $X_i=X_j\cup\{v\}$ for some vertex $v\in V$ and are said to \emph{introduce} $v$;
 \item Forget nodes $i$ have one child $j$ with $X_i=X_j\setminus\{v\}$ for some vertex $v\in V$ and are said to \emph{forget} $v$;
 \item Join nodes $i$ have two children denoted by $i-1$ and $i-2$, with $X_i=X_{i-1}=X_{i-2}$.
\end{enumerate}
Nice tree decompositions were introduced by Kloks in \cite{Kloks94} and using them does not in general give any additional algorithmic possibilities, yet algorithm design becomes considerably easier.

Pathwidth is similarly defined and the only difference in the above definitions is that trees are restricted to being paths. Additionally, we will require the equivalent definition of pathwidth via the \emph{mixed search number} $\ms(G)$. In a \emph{mixed search game}, a graph $G$ is considered as a system of tunnels. Initially, all edges are contaminated by a gas and an edge is \emph{cleared} by placing searchers at both its endpoints simultaneously or by sliding a searcher along the edge. A cleared edge is re-contaminated if there is a path from a contaminated edge to the cleared edge without any searchers on its vertices or edges. A search is a sequence of operations that can be of the following types: (a) placement of a new searcher on a vertex; (b) removal of a searcher from a vertex; (c) sliding a searcher on a vertex along an incident edge and placing the searcher on the other end. A search strategy is winning if after its termination all edges are cleared. The mixed search number of $G$, denoted by $\ms(G)$, is the minimum number of searchers required for a winning strategy of mixed searching on $G$.

\begin{lemma}\label{lem:mixed_search} \cite{TAKAHASHI1995253} For a graph $G$, it is $\pw(G)\le\ms(G)\le\pw(G)+1$. \end{lemma}

We will also use the parameters
\emph{vertex cover number} and \emph{feedback vertex set number}  of a graph
$G$, which are the sizes of the minimum vertex set whose deletion leaves the
graph edgeless, or acyclic, respectively. Finally, we will consider the related
notion of \emph{tree-depth} \cite{NesetrilM06}, which is defined as the minimum
height of a rooted forest whose completion (the graph obtained by connecting
each node to all its ancestors) contains the input graph as a subgraph. We will
denote these parameters for a graph $G$ as
$\tw(G),\pw(G),\vc(G),\fvs(G)$, and $\td(G)$, and will omit $G$ if it is
clear from the context. We recall the following well-known relations between these parameters which justify the hierarchy given in Figure \ref{fig:param_relations}:

\begin{lemma}\label{lem:relations} \cite{BodlaenderGHK95,CourcelleO00} For any
graph $G$ we have $\tw(G) \le \pw(G) \le \td(G) \le \vc(G)$, $\tw(G)\le \fvs(G)
\le vc(G)$.  \end{lemma}

The \textsc{Set Packing} problem is defined as follows: given an integer $k$, a \emph{universe} $\mathcal{U}=\{u_1,\dots,u_n\}$ of \emph{elements} and a family $\mathcal{S}=\{S_1,\dots,S_m\}$ of \emph{subsets} of $\mathcal{U}$, is there a subfamily $S\subseteq \mathcal{S}$ of subsets (a \emph{packing}), such that all sets in $S$ are pairwise disjoint, and the size of the packing is $|S|\ge k$?

Finally, \textsc{$k$-Multicolored Independent Set} is a well-known W[1]-complete problem (see \cite{CyganFKLMPPS15}) and is defined as follows: we are given a graph $G=(V,E)$, with $V$ partitioned into $k$ independent sets $V=V_1\uplus \dots\uplus V_k$, $|V_i|=n, \forall i\in[1,k]$, where $E$ only contains edges between vertices of sets $V_i,V_j$ with $i\not=j$ and we are asked to find a subset $S\subseteq V$, such that $G[S]$ forms an independent set and $|S\cap V_i|=1,\forall i\in[1,k]$.

We also recall here the two main complexity assumptions used in this paper
\cite{ImpagliazzoP01,ImpagliazzoPZ01}. The Exponential Time Hypothesis (ETH)
states that 3-\textsc{SAT} cannot be solved in time $2^{o(n+m)}$ on instances
with $n$ variables and $m$ clauses. The Strong Exponential Time Hypothesis
(SETH) states that for all $\epsilon>0$, there exists an integer $q$ such that
$q$-\textsc{SAT} (where $q$ is the maximum size of any clause) cannot be solved in time $O((2-\epsilon)^n)$.

\section{Treewidth: SETH Lower Bound}\label{sec_tw_lb}

In this section we show that for any fixed $d>2$, the existence of
any algorithm for the \dS\ problem of running time $O^*((d-\epsilon)^{\tw})$,
for some $\epsilon>0$, would imply the existence of some algorithm for
\textsc{$q$-SAT} on instances with $n$ variables, of running time $O^*((2-\delta)^n)$, for
some $\delta>0$ and any $q\ge3$.
First, let us briefly summarize the reduction for the SETH
lower bound of $(2-\epsilon)^{\tw}$ for \textsc{Independent Set} from
\cite{LokshtanovMS11a}. The reduction is based on the construction of $n$ paths
(one for each variable) on $2m$ vertices each, conceptually divided into $m$
pairs of vertices (one for each clause), with each vertex signifying assignment
of value 0 or 1 to the corresponding variable. A gadget is introduced for each
clause, connected to the vertex of some path that signifies the assignment to
the corresponding variable that would satisfy the clause. The pathwidth of the graph (and thus also its treewidth) is (roughly) equal to the
number of paths and so a correspondence between a satisfying
assignment and an independent set can be established, meaning an
$O^*((2-\epsilon)^{\tw})$-time algorithm for \textsc{Independent Set} would
imply an $O^*((2-\epsilon)^n)$-time algorithm for \textsc{SAT}, for any
$\epsilon>0$.

Intuitively, the reduction for \textsc{Independent Set} needs to ``embed'' the $2^n$ possible variable assignments into the $2^{\tw}$ states of some optimal dynamic program for the problem, while in our lower bound construction for \dS\ we need to be able to encode these $2^n$ assignments by $d^{\tw}$ states and thus there can be no one-to-one correspondence between a variable and only one vertex in some bag of the tree decomposition (that the optimal dynamic program might assign states to); instead, every vertex included in some bag must carry information about the assignment for a \emph{group} of variables. Furthermore, as now $d>2$, in order to make the converse direction of our reduction to work, we need to make our paths sufficiently long to ensure that any solution will eventually settle into a pattern that encodes a consistent assignment, as the optimal $d$-scattered set may ``cheat'' by not selecting the same vertex from each part of some long path (periodically), a situation that would imply a different assignment for the appearances of the same variable for two different clauses (see also \cite{BorradaileL16} and the SETH-based lower bound for \textsc{Dominating Set} from \cite{LokshtanovMS11a}).

 \paragraph{Clause gadget $\hat{C}$:} We first describe the construction of our clause gadget $\hat{C}$: this gadget has $N$ \emph{input} vertices and its purpose is to only allow for selection of one of these in any $d$-scattered set, along with another, standard selection. Given vertices $v_1,\dots,v_N$, we first make $N$ paths $A_i=(a^1_i,\dots,a^{\lfloor d/2\rfloor-1}_i),\forall i\in[1,N]$ on $\lfloor d/2\rfloor-1$ vertices. We connect vertices $a^1_i$ to inputs $v_i$, while only for even $d$, we also make all vertices $a^{\lfloor d/2\rfloor-1}_i$ into a clique (all other endpoints of each path). We then make a path $B=(b^1,\dots,b^{\lceil d/2\rceil+1})$ and we connect its endpoint $b^{\lceil d/2\rceil+1}$ to all $a^{\lfloor d/2\rfloor-1}_i$.  Observe that any $d$-scattered set can only include one of the input vertices (as the distance between them is $d-1$) and the vertex $b_1$, being the only option at distance $d$ from all inputs.
 
\begin{figure}[htbp]
\centerline{\includegraphics[width=50mm]{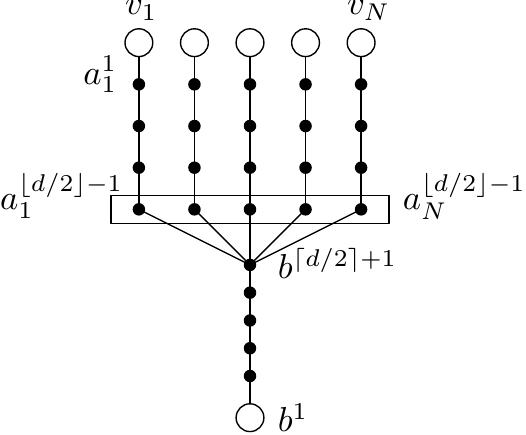}}
\caption{A general picture of the clause gadget $\hat{C}$ for even and odd $d$. Note the box
indicating vertices forming a clique for the case of even $d$.}
\label{fig:clause_gadget2} \end{figure}

 \paragraph{Construction:} We will describe the construction of a graph $G$, given some
 $\epsilon<1,q\ge3,d>2$ and an instance $\phi$ of \textsc{$q$-SAT} with $n$ variables, $m$
 clauses and at most $q$ variables per clause. We first choose an integer $p=\lceil\frac{1}{(1-\lambda)\log_2(d)}\rceil$, for $\lambda=\log_{d}(d-\epsilon)<1$
 (i.e.\ $p$ depends only on $d$ and $\epsilon$) and then group
 the variables of $\phi$ into $t=\lceil\frac{n}{\gamma}\rceil$ groups
 $F_1,\dots,F_t$, for $\gamma=\lfloor\log_2(d)^p\rfloor$, being also the maximum
 size of any such group.
 
 For each group
 $F_{\tau}$ of variables of $\phi$, with $\tau\in[1,t]$, we make a simple gadget
 $\hat{G}_{\tau}^1$ that consists of $p$ paths $P_{\tau}^l=(p_1^l,\dots,p_d^l)$
 on $d$ vertices each, for $l\in[1,p]$. We then make $m(tp(d-1)+1)$ copies of
 this ``column'' of $t$ gadgets $\hat{G}_{1}^1,\dots,\hat{G}_t^1$ (i.e.\ $t$ vertically arranged gadgets), that we
 connect horizontally (so that we have $tp$ ``long paths''): we connect each
 last vertex $p_d^l$ from a gadget $\hat{G}_{\tau}^j$ to vertex $p_1^l$ from the
 following gadget $\hat{G}_{\tau}^{j+1}$, for all $l\in[1,p]$, $\tau\in[1,t]$
 and $j\in[1,m(tp(d-1))]$ (see Figure \ref{fig:global_construction}~(b) for an example).
 
 Next, for every clause $C_{\mu}$, with $\mu\in[1,m]$, we make $tp(d-1)+1$ copies of the clause gadget $\hat{C}^j$, for $j\in[1,m(tp(d-1)+1)]$, where for each $\mu\in[1,m]$, the number of inputs in the $tp(d-1)+1$ copies is $N=q_{\mu}d^p/2$, where $q_{\mu}$ is the number of literals in clause $C_{\mu}$. One clause is assigned to each column of gadgets, so that the first $m$ columns correspond to one clause each, with $tp(d-1)+1$ repetitions of this pattern giving the complete association. Then, for every $\tau\in[1,t]$ we associate a set $S_{\tau}\subset\bigcup_{l\in[1,p]} P_{\tau}^l$, that contains exactly one vertex from each of the $p$ paths in $\hat{G}_{\tau}^j$, with an assignment to the variables in group $F_{\tau}$. As there are at most $2^{\gamma}=2^{\lfloor\log_2(d)^p\rfloor}$ assignments to the variables in $F_{\tau}$ and $d^p\ge2^{\gamma}$ such sets $S_{\tau}$, the association can be unique for each $\tau$ (i.e.\ for each \emph{row} of gadgets). Now, for every literal appearing in clause $C_{\mu}$, exactly half of the partial assignments to the group $F_{\tau}$ in which the literal's variable appears will satisfy it and thus, each of the $q_{\mu}d^p/2$ input vertices of the clause gadget will correspond to one literal and one assignment to the variables of the group that satisfy it.
 
 Let $v$ be an input vertex of a clause gadget $\hat{C}^j$, corresponding to a literal of clause $C_{\mu}$ that is satisfied by a partial assignment to the variables of group $F_{\tau}$ that is associated with set $S_{\tau}\subset \bigcup_{l\in[1,p]} P_{\tau}^l$, containing exactly one vertex from each path $P_{\tau}^l,l\in[1,p]$, from gadget $\hat{G}_{\tau}^j$. For even $d$, we then make a path $w_1,\dots,w_{d/2-1}$ on $d/2-1$ vertices, connecting vertex $w_1$ to $v$ and for each vertex $p_i^l\notin S$ of each path $P_{\tau}^l\in\hat{G}_{\tau}^j$ we also make a path $y_1,\dots,y_{d/2-1}$ on $d/2-1$ vertices, attaching endpoint $y_1$ to its corresponding path vertex $p_i^l$, while the other endpoints $y_{d/2-1}$ are all attached to vertex $w_{d/2-1}$ and to each other (into a clique). For odd $d$, we make a similar construction for each such $v$, only the number of vertices in constructed paths is now $\lfloor d/2\rfloor$ instead of $d/2-1$ and vertices $y_{\lfloor d/2\rfloor}$ are not made into a clique. Thus every input vertex $v$ of some clause gadget is at distance exactly $d-1$ from every path vertex that does not belong to the set associated with its corresponding partial assignment (and thus exactly $d$ from the only vertex per path that is), while the distances between any pair of other (i.e.\ intermediate) vertices via these paths are $\le d-1$. This concludes our construction, while Figure \ref{fig:global_construction} provides illustrations of the above.

In this way, a satisfying assignment for $\phi$ would correspond to a $d$-scattered set that selects the vertices in each gadget $\hat{G}_{\tau}$ that match the partial assignment $S_{\tau}$ for that group's variables $F_{\tau}$ in all $m(tp(d-1)+1)$ columns, along with the corresponding input vertex from each clause gadget (implying the existence of a satisfied literal within the clause). On the other hand, for any $d$-scattered set of size $(tp+2)m(tp(d-1)+1)$ in $G$, the maximum number of times it can ``cheat'' by not periodically selecting the ``same'' vertices in each column is $tp(d-1)$. The number of columns being $m(tp(d-1)+1)$, by the pigeonhole principle, there will always exist $m$ consecutive columns for which the selection pattern does not change, from which a consistent assignment for all clauses can be extracted.

\begin{figure}
\centering 
\subfloat[]
{\includegraphics[width=40mm]{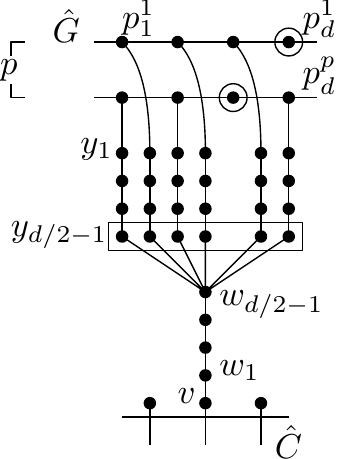}}\qquad\qquad
\subfloat[]
{\includegraphics[width=60mm]{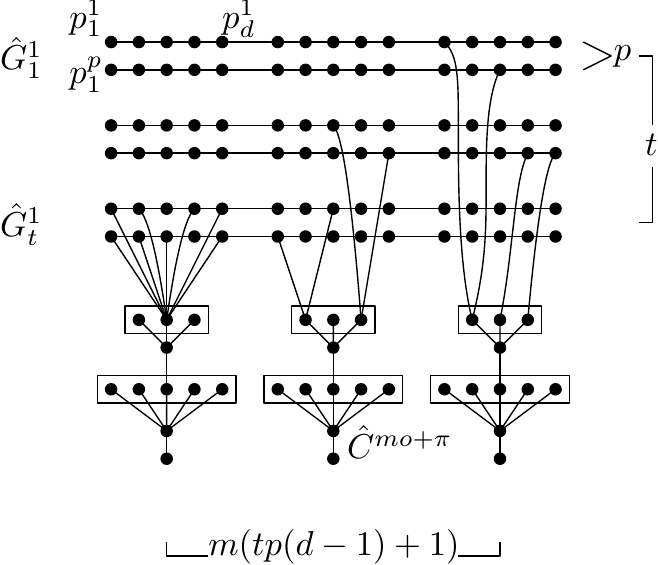}}
\caption{(a): The connection of an input vertex $v$ of a clause gadget $\hat{C}$ to its corresponding path vertices in some $\hat{G}$, where vertices of set $S_{\tau}$ are circled and boxed vertices form a clique (for even $d$). (b): A simplified picture of the global construction, with some exemplative connecting paths between clause gadgets and path vertices shown as edges.}
\label{fig:global_construction}
\end{figure}

\begin{lemma}\label{tw_LB_FWD}
If $\phi$ has a satisfying assignment, then $G$ has a $d$-scattered set of size $(tp+2)m(tp(d-1)+1)$.
\end{lemma}
\begin{proof}
 Given a satisfying assignment for $\phi$, we will show the existence of a $d$-scattered set $K$ of $G$ of size $|K|=(tp+2)m(tp(d-1)+1)$. Set $K$ will include one vertex from each of the $p$ paths in each gadget $\hat{G}_{\tau}^{j}$, with $\tau\in[1,t],j\in[1,m(tp(d-1)+1)]$ and two vertices from each clause gadget $\hat{C}^j$. In particular, for each group $F_{\tau}$ of variables we consider the restriction of the assignment for $\phi$ to these variables and identify the set $S_{\tau}$ associated with this partial assignment, adding all vertices of $S_{\tau}$ from each $\hat{G}_{\tau}^j$ into set $K$. Then, for each clause $C_{\mu}$, we identify one satisfied literal (which must exist as the assignment for $\phi$ is satisfying) and the vertex $v$ that corresponds to that literal and the partial assignment associated with set $S_{\tau}$ selected from the group of paths $P_{\tau}^l,l\in[1,p]$ within gadget $\hat{G}_{\tau}^j$, for group $F_{\tau}$, in which that literal's variable appears in. We add to $K$ every such vertex $v$ and also vertex $b_1$ from each clause gadget $\hat{C}^j$, for all $j\in[1,m(tp(d-1)+1)]$, thus completing the selection and what remains is to show that $K$ is indeed a $d$-scattered set of $G$.
 
 To that end, observe that the pattern of our $p\cdot m(tp(d-1)+1)$ selections of all vertices from every set $S_{\tau}$ for each $\tau\in[1,t]$ is repeating: we have selected every $d$-th vertex from each ``long path'', since the association between sets $S_{\tau}$ and partial assignments for $F_{\tau}$ is the same of each $\tau$. Thus on each of the $tp$ long paths, every selected vertex is at distance exactly $d$ both from its predecessor and its follower. Furthermore, for each clause gadget $\hat{C}^j$, with $j\in[1,m(tp(d-1)+1)]$, selected vertices $v$ and $b_1$ are at distance exactly $d$ via the gadget, while vertex $v$ is at distance exactly $d$ from each selected $p_i^l\in S_{\tau}$ from each path $P_{\tau}^l$, $l\in[1,p]$, as there are only paths of length $d-1$ from $v$ to the neighbors $p_{i-1}^l,p_{i+1}^l$ of the selected vertex from each path. Finally, observe that the distance between vertices on different paths $P_{\tau}^l$ (and thus possible selections) via the paths attached to some input vertex is always $\ge2d-2$. 
\end{proof}

\begin{lemma}\label{tw_LB_BWD}
 If $G$ has a $d$-scattered set of size $(tp+2)m(tp(d-1)+1)$, then $\phi$ has a satisfying assignment.
\end{lemma}
\begin{proof}
 Given a $d$-scattered set $K$ of $G$ of size $|K|=(tp+2)m(tp(d-1)+1)$, we will show the existence of a satisfying assignment for $\phi$. First, observe that from each gadget $\hat{G}_{\tau}^j$, for $\tau\in[1,t],j\in[1,m(tp(d-1)+1)]$, at most $p$ vertices can be selected, one from each path $P_{\tau}^l,l\in[1,p]$ within each gadget. This leaves 2 vertices to be selected from each column $j\in[1,m(tp(d-1)+1)]$ of gadgets. As the distances between some input vertex $v$ and some path vertex $p_i^l$ is equal to $d$ only if the path vertex belongs to the set $S_{\tau}$ associated with the partial assignment to the variables of $F_{\tau}$ that would satisfy the literal (whose variable belongs to $F_{\tau}$) corresponding to the input vertex $v$ and $d-1$ otherwise, while the distances between any pair of input vertices are $d-1$ via the gadget with only vertex $b_1$ at distance exactly $d$ from each input vertex, it is not hard to see that the only option is to select each vertex $b_1$ and one input vertex from each gadget $\hat{C}^j$, for $j\in[1,m(tp(d-1)+1)]$: no path vertex $p_i^l$ could be selected with any vertex on the paths attached to some input vertex $v$, while no other vertex but $b_1$ could be selected with some input vertex of each clause gadget. Furthermore, the selection of an input vertex $v$ must also be in agreement with each selection from the $p$ paths to which $v$ is connected to (via the paths of length $d-1$), i.e.\ the selected vertices from each path must be exactly the set $S_{\tau}$ that is associated with the partial assignment that satisfies the literal corresponding with $v$.
 
 Next, we require that there exists at least one $o\in[0,tp(d-1)]$ for every $\tau\in[1,t]$ for which $K\cap\{\bigcup_{l\in[1,p]}P_{\tau}^l\}$ is the same in all gadgets $\hat{G}_{\tau}^{mo+\pi}$ with $\pi\in[1,m]$, i.e.\ that there exist $m$ successive copies of the gadget for which the pattern of selection of vertices from paths $P_{\tau}^l$ does not change. As noted above, set $K$ must contain one vertex from each such path in each gadget, while the distance between any two successive selections (on the same ``long path'') must be at least $d$. Now, depending on the starting selection, observe that the pattern can ``shift towards the right'' at most $d-1$ times, without affecting whether the total number of selections is exactly $m(tp(d-1)+1)$ from each ``long path'': the first vertex of a path can be selected within a gadget, the second vertex can be selected from its follower, the third from the one following it and so on. For each $l\in[1,p]$, this can happen at most $d-1$ times, thus at most $p(d-1)$ for each $\tau\in[1,t]$ and $tp(d-1)$ over all $\tau$. By the pigeonhole principle, there must thus exist an $o\in[0,tp(d-1)]$ such that no such shift happens among the gadgets $\hat{G}_{\tau}^{mo+\pi}$, for all $\tau\in[1,t]$ and $\pi\in[1,m]$.
 
 Our assignment for $\phi$ is then given by the selections for $K$ in each gadget $\hat{G}_{\tau}^{mo+1}$ for this $o$: for every group $F_{\tau}$ we consider the selection of vertices from $P_{\tau}^l$ for $l\in[1,p]$, forming subset $S_{\tau}$ and its associated partial assignment to the variables of $F_{\tau}$. In this way we get an assignment to all the variables of $\phi$. To see why this also satisfies every clause $C_{\pi}$ with $\pi\in[1,m]$, consider clause gadget $\hat{C}^{mo+\pi}$: there must be an input vertex $v$ selected from this gadget, corresponding to a satisfying partial assignment for some literal of $C_{\pi}$, that must be at distance exactly $d$ from each path selection that together give subset $S_{\tau}$, the subset associated with this satisfying partial assignment. 
 \end{proof}
 
 \begin{lemma}\label{tw_LB_bound}
 Graph $G$ has treewidth $\tw(G)\le tp+qd^p/2+d$.
\end{lemma}
\begin{proof}
 We will in fact show a pathwidth bound of $\pw(G)\le tp+qd^p/2+d$ by providing a mixed strategy to clean $G$ using $tp+qd^p/2+d$ searchers. The claimed bound on the treewidth then follows from lemmas \ref{lem:relations} and \ref{lem:mixed_search}.
 
 We initially place one searcher on every first vertex $p_1^l$ of every path $P_{\tau}^l$ in each gadget $\hat{G}_{\tau}^1$ for all $l\in[1,p]$ and $\tau\in[1,t]$. We also place a searcher on vertex $b_1$ of clause gadget $\hat{C}^1$, also one on each of its $q_{\mu}d^p/2$ vertices $a_i^{\lfloor d/2\rfloor-1}$ (between the inputs and $b_1$) and finally one searcher on each of the $d-1$ vertices $y_{d/2-1}$ (or $y_{\lfloor d/2\rfloor}$ for odd $d$) that are connected through a $w_1,\dots,w_{d/2-1}$ path to each input vertex (or $w_1,\dots,w_{\lfloor d/2\rfloor}$). 
 
 We then slide the searcher on $b_1$ over the path $b_1,\dots,b^{\lceil d/2\rceil}$ until all the path's edges as well as the edges between $b^{\lceil d/2\rceil}$ and every $a_i^{\lfloor d/2\rfloor-1}$ are cleaned (the clique edges between the $a_i^{\lfloor d/2\rfloor-1}$ for even $d$ are also clean). We then slide the searchers from the $a_i^{\lfloor d/2\rfloor-1}$ along each path to each input vertex and from there on along the paths $w_1,\dots,w_{d/2-1}$ (or $w_{\lfloor d/2\rfloor}$ for odd $d$). In this way all these paths and the edges between the $w_{d/2-1}$ and $y_{d/2-1}$ (or $w_{\lfloor d/2\rfloor}$ and $y_{\lfloor d/2\rfloor}$ for odd $d$) are cleaned and we can slide the searchers from each $y_{d/2-1}$ down to each $y_1$ (being adjacent to one path vertex each). We then slide all $tp$ searchers from the first vertices $p_1^l$ along their paths $P_{\tau}^l$ for $l\in[1,p]$ in each gadget $\hat{G}_{\tau}^1$. After all edges of the first column have been cleaned in this way, we slide the $tp$ searchers on the first vertices of each path of the following column, we remove the searchers from the vertices of the clause gadget (and adjacent paths) and place them on their corresponding starting positions on the following column. We then repeat the above process until all columns have been cleaned. We thus use at most $tp+qd^p/2+d$ searchers simultaneously, where $qd^p/2+d=O(1)$. 
\end{proof}

\begin{theorem}\label{tw_LB_THM}
 For any fixed $d>2$, if \dS\ can be solved in $O^*((d-\epsilon)^{\tw(G)})$ time for some $\epsilon>0$, then there exists some $\delta>0$, such that \textsc{$q$-SAT} can be solved in $O^*((2-\delta)^n)$ time, for any $q\ge3$.
\end{theorem} 
\begin{proof}
 Assuming the existence of some algorithm of running time $O^*((d-\epsilon)^{\tw(G)})=O^*(d^{\lambda\tw(G)})$ for \dS, where $\lambda=\log_{d}(d-\epsilon)$, we construct an instance of \dS\, given a formula $\phi$ of \textsc{$q$-SAT}, using the above construction and then solve the problem using the $O^*((d-\epsilon)^{\tw(G)})$-time algorithm. Correctness is given by Lemma \ref{tw_LB_FWD} and Lemma \ref{tw_LB_BWD}, while Lemma \ref{tw_LB_bound} gives the upper bound on the running time:
 
 \begin{align}
  O^*(d^{\lambda\tw(G)})&\le O^*\left(d^{\lambda(tp+f(d,\epsilon,q))}\right)\\
  &\le O^*\left(d^{\lambda p\left\lceil\dfrac{n}{\lfloor\log_2(d)^p\rfloor}\right\rceil}\right)\label{tw_comp_2}\\
  &\le O^*\left(d^{\lambda p\dfrac{n}{\lfloor\log_2(d)^p\rfloor}+\lambda p}\right)\\
  &\le O^*\left(d^{\lambda\dfrac{np}{\lfloor p\log_{2}(d)\rfloor}}\right)\label{tw_comp_4}\\
  &\le O^*\left(d^{\delta'\dfrac{n}{\log_2(d)}}\right)\label{tw_comp_5}\\
  &\le O^*(2^{\delta''n})=O((2-\delta)^n)
 \end{align}
for some $\delta,\delta',\delta''<1$. Observe that in line (\ref{tw_comp_2}) the function $f(d,\epsilon,q)=qd^p/2+d$ is considered constant, as is $\lambda p$ in line (\ref{tw_comp_4}), while in line (\ref{tw_comp_5}) we used the fact that there always exists a $\delta'<1$ such that $\lambda\dfrac{p}{\lfloor p\log_2(d)\rfloor}=\dfrac{\delta'}{\log_2(d)}$, as we have:
\begin{equation*}
 \begin{split}
    p\log_2(d)-1&<\lfloor p\log_2(d)\rfloor\\
    \Leftrightarrow\dfrac{\lambda p\log_2(d)}{p\log_2(d)-1}&>\dfrac{\lambda p\log_2(d)}{\lfloor p\log_2(d)\rfloor},\\
    \text{from which, by substitution, we get } \dfrac{\lambda p\log_2(d)}{p\log_2(d)-1}&>\delta',\\
    \text{now requiring } \dfrac{\lambda p\log_2(d)}{p\log_2(d)-1}&\le1,\\
    \text{or } p\ge\dfrac{1}{(1-\lambda)\log_2(d)},
 \end{split}
\end{equation*}
that is precisely our definition of $p$. This concludes the proof. 
\end{proof}

\section{Treewidth: Dynamic Programming Algorithm}\label{sec_tw_dp}
We present an $O^*(d^{\tw})$-time algorithm for the counting version of the \dS\ problem.
 The input is a graph $G=(V,E)$, a nice tree decomposition
$(\mathcal{X},T)$ for $G$, where $T=(I,F)$ is a tree and
$\mathcal{X}=\{X_i|i\in I\}$ is the set of bags, while $\max_{i\in
i}|X_i|-1=\tw$, along with two numbers $k\in\mathbb{N}^+,d\ge 2$, while the
output is the number of $d$-scattered sets of size $k$ in $G$.

\paragraph{Table description:} There is a table $D_i$ associated with every node $i\in I$ of the tree decomposition with $X_i=\{v_0,\dots,v_t\},0\le t\le\tw$, while each table entry $D_i[\kappa,s_0,\dots,s_t]$ contains the number of (disjoint) $d$-scattered sets $K\subseteq V_i$ of size $|K|=\kappa$ (its \emph{partial solution}) and is indexed by a number $\kappa\in[1,k]$ and a $t+1$-sized tuple $(s_0,\dots,s_t)$ of \emph{state-configurations}, assigning a state $s_j\in[0,d-1]$ to each vertex $v_j,\forall j\in[0,t]$. There are $d$ possible states for each vertex, designating its distance to the closest selection at the ``current'' stage of the algorithm:
\begin{itemize}
 \item \emph{Zero} state $s_j=0$ signifies vertex $v_j$ is considered for selection in the $d$-scattered set and is at distance at least $d$ from any other such selection: $\forall u\in K: d(u,v_j)\ge d$.
 \item \emph{Low} states $s_j\in[1,\lfloor d/2\rfloor]$ signify vertex $v_j$ is at distance at least $s_j$ from its closest selection and at least $d-s_j$ from the second closest: $\forall u,w\in K|d(u,v_j)\le d(w,v_j): d(u,v_j)\ge s_j\wedge d(w,v_j)\ge d-s_j$.
 \item \emph{High} states $s_j\in[\lfloor d/2\rfloor+1,d-1]$ signify vertex $v_j$ is at distance at least $s_j$ from its closest selection: $\forall u\in K:d(u,v_j)\ge s_j$.
\end{itemize}

For a node $i\in I$, each table entry $D_i[\kappa,s_0,\dots,s_t]$ contains the number of $d$-scattered sets $K\subseteq V_i$ of the terminal subgraph $G_i$, such that the situation of each vertex in the corresponding bag is being described by the particular state configuration $(s_0,\dots,s_t)$ indexing this entry. The computation of each entry is based on the type of node the table is associated with (leaf, introduce, forget, or join), previously computed entries of the table associated with the preceding node(s) and the structure of the node's terminal subgraph. In particular, we have $\forall i\in I, D_i[\kappa,s_0,\dots,s_t]:[1,k]\times[0,d-1]^{t+1}\mapsto\mathbb{N}^0$, where $0\le t\le\tw$. The inductive computation of all table entries for each type of node follows.

\paragraph{Leaf node $i$ with $X_i=\{v_0\}$:}
\begin{equation*}
D_i[\kappa,s_0] \coloneqq
\begin{cases}
 1			& \mbox{, if } s_0=0, \kappa=1; \\
 1			& \mbox{, if } s_0>0, \kappa=0; \\
 0			& \mbox{, otherwise.}
\end{cases}
\end{equation*}
Leaf nodes contain only one vertex and there is one $d$-scattered set that includes this vertex for ($s_0=0$, $\kappa=1$) and one $d$-scattered set that does not (for $s_0>0$, $\kappa=0$).

\paragraph{Introduce node $i$ with $X_i=X_{i-1}\cup\{v_{t+1}\}$:}

\begin{equation*}
 D_i[\kappa,s_0,\dots,s_t,s_{t+1}] \coloneqq
 \begin{cases}
 D_{i-1}[\kappa,s_0,\dots,s_t], & \mbox{if } s_{t+1}\in[1,d-1] \mbox{ and} \\
		     & s_{t+1}\le\min_{v_j\in X_{i-1}}(d(v_{t+1},v_j)+s_j)  ;\\
 D_{i-1}[\kappa',s'_0,\dots,s'_t], & \mbox{if } s_{t+1}=0, \kappa'=\kappa-1 \mbox{ and } \forall v_j\in X_{i-1} \mbox{ with } s_j=0,\\
			   & \mbox{it is } d(v_{t+1},v_j)\ge d, \mbox{ and }\\
			   & \forall v_j\in X_{i-1} \mbox{ with } d(v_{t+1},v_j)\le \lfloor d/2\rfloor,\\
			   & \mbox{it is }  s_j\le d(v_{t+1},v_j) \mbox{ and } s'_j= d-s_j,\\
		           & \mbox{with } s'_j=s_j, \mbox{ if } d(v_{t+1},v_j)>\lfloor d/2\rfloor;\\
 0,                  & \mbox{otherwise.}
 \end{cases}
\end{equation*}
When a vertex is introduced, for previously computed partial solutions to be correctly extended, we require that its given state matches the distance/state conditions of the other vertices in the bag, while if the introduced vertex is considered for selection, then the previous entries we examine must ensure this selection is possible.

\paragraph{Forget node $i$ with $X_i=X_{i-1}\setminus\{v_{t+1}\}$:}

\begin{equation*}
  D_i[\kappa,s_0,\dots,s_t]\coloneqq \sum_{s_{t+1}\in[0,d-1]}\{D_{i-1}[\kappa,s_0,\dots,s_t,s_{t+1}]\}.
\end{equation*}
The correct value for each entry is the sum over all states of the forgotten vertex $v_j$, where the size of the $d$-scattered sets is $\kappa$.\footnote{We remark that only in the case of a forget node following a leaf node for some vertex $v_0$, the algorithm does not compute the sum over all states of the forgotten vertex as this would give a value of $d-1$, but only of states $s_0\in\{0,1\}$ for a correct value of 2.}

\paragraph{Join node $i$ with $X_i=X_{i-1}=X_{i-2}$:} Given state-tuple $(s_0,\dots,s_t)$, we assume (without loss of generality) that for some $t'\in[0,t]$ (if any) it is $s_j\in[1,\lfloor d/2\rfloor],\forall j\in[0,t']$ and also $\min_{\forall v_l\in X_i|s_l=0} d(v_j,v_l)>\lfloor d/2\rfloor$, while all other vertices are $v_{t'+1},\dots,v_t$. Now, let $\mathcal{S}_{t'}^{\le}$ denote the set of all possible tuples $S=(s_0^{\le},\dots,s_{t'}^{\le})$, where each state $s_j^{\le}$ is either the same state $s_j$, or its symmetrical (around $d/2$):
\begin{equation*}
 \mathcal{S}_{t'}^{\le}\coloneqq\{(s_0^{\le},\dots,s_{t'}^{\le})|\forall j\in[0,t']:\{(s_j^{\le}=s_j)\vee(s_j^{\le}=d-s_j)\}\},
\end{equation*}
while for some tuple $S\in\mathcal{S}_{t'}^{\le}$, let $\bar{S}$ denote the complementary tuple (where the state of each vertex is likewise reversed) and also let $\kappa''$ denote the number of zero states in $(s_{t'+1},\dots,s_t)$. Then we have:
\begin{equation*}
 D_i[\kappa,s_0,\dots,s_t] \coloneqq \sum_{S\in\mathcal{S}_{t'}^{\le}}\{D_{i-1}[\kappa',S,s_{t'+1},\dots,s_t]\cdot D_{i-2}[\kappa-\kappa'+\kappa'',\bar{S},s_{t'+1},\dots,s_t]\}.
\end{equation*}
For join nodes, the bags of both children contain the same set of vertices, yet the partial solutions characterized by the entries of each table concern distinct terminal subgraphs $G_{i-1}$ and $G_{i-2}$. For state-configurations where some vertices are of low state (that is not justified by the presence of some vertex of zero state within the bag), the closest selection to these vertices (that gives the state) might be in any of the two terminal subgraphs, but not both: if the ``target'' state is $s_j\in[1,\lfloor d/2\rfloor]$, then there might be a selection in $G_{i-1}$ at distance $s_j$ but there must not be another selection in $G_{i-2}$ at distance $\le d-s_j$ (and vice-versa).

\paragraph{State changes:} The computations at a join node as described above would add an additional factor in the complexity of our algorithm if implemented directly, yet this can be avoided by an application of the \emph{state changing} technique (or fast subset convolution, see \cite{BjorklundHKK07,BodlaenderLRV10,RooijBR09} and Chapter 11 from \cite{CyganFKLMPPS15}): since the number of entries involved can be exponential in $\tw$ (due to the size of $\mathcal{S}_{t'}^{\le}$), in order to efficiently compute the table for a join node $i$, we will first transform the tables $D_{i-1},D_{i-2}$ of its children into tables $D_{i-1}^*,D_{i-2}^*$ of a new type that employs a different state representation, for which the join operation can be efficiently performed to produce table $D_i^*$, that we will finally transform back to table $D_i$, thus progressing with our dynamic programming algorithm.

In particular, each entry $D_i^*[\kappa,s_1,\dots,s_t]$ of the new table will be an aggregate of entries $D_i[\kappa,s_0,\dots,s_t]$ of the original table, with its value equal to the sum of the appropriate values of the corresponding entries. For vertex $v_j$, each low state $s_j\in[1,\lfloor d/2\rfloor]$ in the new state signification for table $D_i^*$ that is not justified by the presence of an appropriate selection within the bag (i.e.\ its minimum distance to any zero-state vertex is at least $\lfloor d/2\rfloor+1$) will correspond to both the same low state $s_j$ and its symmetrical high state $d-s_j$ from the original signification. Observe that these correspondences exactly parallel the definition of set $\mathcal{S}_{t'}^{\le}$ used in the original computations.

First, let $D_i^*$ be a copy of table $D_i$. The transformation then works in $t$ steps, vertex-wise: we require that all entries $D_i^*[\kappa,s_0^*,\dots,s^*_t]$ contain the sum of all entries of $D_i$ where for low states $s_j^*$ (that are also not justified by some present selection), it is $s_j^*=s_j$ or $s_j^*=d-s_j$, and all other vertex-states and $\kappa$ are fixed: at step $j$, we add $D_i^*[\kappa,s_0,\dots,s_j,\dots,s_t]=D_i[\kappa,s_0,\dots,s_j,\dots,s_t]+D_i[\kappa,s_0,\dots,d-s_j,\dots,s_t]$ if $s_j\in[1,\lfloor d/2\rfloor]$ and $\min_{\forall v_l\in X_i|s_i=0}d(v_j,v_l)>\lfloor d/2\rfloor$. We then proceed to the next step for $v_{j+1}$ until table $D_i^*$ is computed. Observe that the above procedure is fully reversible:\footnote{This is the reason for \emph{counting} the number of solutions for each $\kappa$: there is no additive inverse operation for the max-sum semiring, yet the sum-product ring is indeed equipped with subtraction.} to invert table $D_i^*$ back to table $D_i$, we again work in $t$ steps, vertex-wise: we first let $D_i$ be a copy of $D_i^*$ and then at step $j$ for all other vertex-states and $\kappa$ fixed, we subtract $D_i[\kappa,s_0,\dots,s_j,\dots,s_t]=D_i^*[\kappa,s_0,\dots,s_j,\dots,s_t]-D_i^*[\kappa,s_0,\dots,d-s_j,\dots,s_t]$ if $s_j\in[1,\lfloor d/2\rfloor]$ and $\min_{\forall v_l\in X_i|s_i=0}d(v_j,v_l)>\lfloor d/2\rfloor$. For both transformations, we perform at most one addition per $k\cdot d^{t+1}/2$ entries for each step $j\in[0,t]$.

Thus we can compute table $D_i^*$ by simply multiplying the values of the two corresponding entries from $D_{i-1}^*,D_{i-2}^*$, as they now contain all required information for this state representation, with the inverse transformation of the result giving table $D_i$:
\begin{equation*}
 D_i^*[\kappa,s_0,\dots,s_t]\coloneqq\sum_{\kappa'=0}^{\kappa'=\kappa}D_{i-1}^*[\kappa',s_0,\dots,s_t]\cdot D_{i-2}^*[\kappa-\kappa'+\kappa'',s_0,\dots,s_t].
\end{equation*}

\begin{theorem}\label{tw_DP_THM}
 Given graph $G$, along with $d\in\mathbb{N^+}$ and nice tree decomposition $(\mathcal{X},T)$ of width $\tw$ for $G$, there exists an algorithm to solve the counting version of the \dS\ problem in $O^*(d^{\tw})$ time.
\end{theorem}
\begin{proof}
 
Let $U_i(\kappa,s_0,\dots,s_t)=\{K\subseteq V_i|K\cap X_i= \{v_j\in\ X_i|s_j=0\},|K|=\kappa,\forall u,v\in K: d(u,v)\ge d\}$ be the set of all $d$-scattered sets in $G_i$ of size $\kappa$. To show correctness of our algorithm we need to establish that for every node $i\in I$, each table entry $D_i[\kappa,s_0,\dots,s_t]$ contains the size of a partial solution to the problem as restricted to $G_i$, i.e.\ the size of this set $|U_i(\kappa,s_0,\dots,s_t)|$, such that the distance between every pair of vertices in $K$ is at least $d$, while for every vertex $v_j\in X_i$, its state for this entry describes its situation within this partial solution. In particular, we need to show the following:

\begin{gather}
 \forall i\in I, \forall \kappa\in[1,k],\forall (s_0,\dots,s_t)\in[0,d-1]^{t+1},0\le t<\textrm{tw}:\label{tw_cor_1}\\
 D_i[\kappa,s_0,\dots,s_t]=|U_i(\kappa,s_0,\dots,s_t)|: \label{tw_cor_2}\\
 \left\{\forall u,w\in K:d(u,w)\ge d\right\}\wedge\label{tw_cor_3}\\
 \wedge\{\forall v_j\in X_i|s_j\in[1,\lfloor d/2\rfloor],\forall u,w\in K|d(u,v_j)\le d(w,v_j):\label{tw_cor_4}\\
 d(u,v_j)\ge s_j\wedge d(w,v_j)\ge d-s_j \}\wedge \label{tw_cor_5}\\
 \wedge \{\forall v_j\in X_i|s_j\in[\lfloor d/2\rfloor+1,d-1],\forall u\in K: d(u,v_j)\ge s_j\}\label{tw_cor_6}.
\end{gather}

In words, the above states that for every node $i$, $\kappa\in[1,k]$ and all possible state-configurations $(s_0,\dots,s_t)$ (\ref{tw_cor_1}), table entry $D_i[\kappa,s_0,\dots,s_t]$ contains the size of set $U_i(\kappa,s_0,\dots,s_t)$ containing all subsets $K$ of $V_i$ (that include all vertices $v_j\in X_i$ of state $s_j=0$) of size $|K|=\kappa$ (\ref{tw_cor_2}), such that the distance between every pair of vertices $u,w\in K$ is at least $d$ (\ref{tw_cor_3}), for every vertex $v_j\in X_i$ with low state $s_j\in[1,\lfloor d/2\rfloor]$ and a pair of vertices $u,w$ from $K$ with $u$ closer to $v_j$ than $w$ (\ref{tw_cor_4}), its distance to $u$ is at least equal to its state $s_j$, while its distance to $w$ is at least $d-s_j$ (\ref{tw_cor_5}), while for every vertex $v_j\in X_i$ with high state $s_j\in[\lfloor d/2\rfloor+1,d-1]$ and a vertex $u$ from $K$, its state $s_j$ is at most its distance to $u$ (\ref{tw_cor_6}). This is shown by induction on the nodes $i\in I$:
\begin{itemize}
 \item Leaf node $i$ with $X_i=\{v_0\}$: This is the base case of our induction. There is only one $d$-scattered set $K$ in $V_i$ of size $\kappa=1$, for which (\ref{tw_cor_3}-\ref{tw_cor_6}) is true, that includes $v_0$ and only one for $\kappa=0$ that does not. In the following cases, we assume (our inductive hypothesis) that all entries of $D_{i-1}$ (and $D_{i-2}$ for join nodes) contain the correct number of sets $K$.
 \item Introduce node $i$ with $X_i=X_{i-1}\cup\{v_{t+1}\}$: For entries with $s_j\in[1,\lfloor d/2\rfloor]$, validity of (\ref{tw_cor_3},\ref{tw_cor_6}) is not affected, while for (\ref{tw_cor_4}-\ref{tw_cor_5}): it is $s_{t+1}\le d(v_{t+1},v_j)+s_j$ for some vertex $v_j\in X_{i-1}$, for which, by the induction hypothesis we have that $s_j\le d(u,v_j)$ and $d-s_j\le d(w,v_j)$, where $u$ is the closest selection to $v_j$ and $w$ the second closest. To see the same holds for $v_{t+1}$, observe that $s_{t+1}\le d(v_{t+1},v_j)+d(u,v_j)\le d(v_{t+1},u)$ (by substitution) and $d-s_j\le d(w,v_j)\Rightarrow d-d(w,v_j)\le s_j\Rightarrow d-d(w,v_j)+d(v_j,v_{t+1})\le d(v_{t+1},v_j)+s_j\Rightarrow d-d(w,v_{t+1})\le s_{t+1}\Rightarrow d-s_{t+1}\le d(w,v_{t+1})$.
 
 For entries with $s_{t+1}\in[\lfloor d/2\rfloor+1,d-1]$, validity of (\ref{tw_cor_3}-\ref{tw_cor_5}) is not affected, while for (\ref{tw_cor_6}): it is $s_{t+1}\le d(v_{t+1},v_j)+s_j$ for some vertex $v_j\in X_{i-1}$, for which, by the induction hypothesis we have that $s_j\le d(u,v_j)$ and thus $s_{t+1}\le d(v_{t+1},v_j)+d(u,v_j)\le d(v_{t+1},u)$.
 
 For entries with $s_{t+1}=0$, observe that the low states $s_j$ of vertices $v_j\in X_{i-1}$ in the new entry with $d(v_{t+1},v_j)\le\lfloor d/2\rfloor$ (for otherwise their situation has not changed by addition of $v_{t+1}$ with $s_{t+1}=0$) would correspond to a high original state $s'_j=d-s_j$ in the previously computed entry, for which partial solution we know that $d(v_j,u)\ge s'_j,\forall u\in K$, or that the previously closest selection was at distance at least $s'_j$ (\ref{tw_cor_4}-\ref{tw_cor_5}). For high states $s_j$ of vertices $v_j\in X_{i-1}$, the requirement is exactly $d(v_{t+1},v_j)\ge s_j$ (\ref{tw_cor_6}) and finally, for (\ref{tw_cor_3}), if there was some $u\in K$ such that $u\notin X_{i-1}$ and $d(u,v_{t+1})<d$, then there must be some $v_j\in X_{i-1}$  (on the path between $u$ and $v_{t+1}$), for which if $d(v_{t+1},v_j)\le\lfloor d/2\rfloor$ and $s_j$ is low then (\ref{tw_cor_6}) was false (as $s'_j$ must have been high and matching $d(u,v_j)$), while if $d(v_{t+1},v_j)>\lfloor d/2\rfloor$ and $s_j$ is high, then (\ref{tw_cor_4}-\ref{tw_cor_5}) was false (in all other cases it would not be $d(u,v_{t+1})<d$).
 \item Forget node $i$ with $X_i=X_{i-1}\setminus\{v_{t+1}\}$: In a forget node, the only difference for the partial solutions in which the forgotten vertex was of state $s_{t+1}=0$ is that now vertex $v_{t+1}$ is included in set $K$ only and not $X_i$. Thus, due to (\ref{tw_cor_3}), the correct number is indeed the sum over all states for $v_{t+1}$.
 \item Join node $i$ with $X_i=X_{i-1}=X_{i-2}$: Observe that for (\ref{tw_cor_3}), if there was a pair $u,w\in K\cap X_i$ or $u\in K\cap X_i,w\in K\setminus X_i$ at $d(u,w)<d$, then (\ref{tw_cor_3}) was not true for either $i-1$ or $i-2$, while if there was a pair $u\in K\cap V_{i-1}\setminus X_{i-1},w\in K\cap V_{i-2}\setminus X_{i-2}$ with $d(u,w)<d$, then there must be some vertex $v_j\in X_i$ (on the path between the two) for which (\ref{tw_cor_5}) was not true. For (\ref{tw_cor_4}-\ref{tw_cor_6}), observe that for vertices $v_j$ of low state $s_j$, (\ref{tw_cor_4}-\ref{tw_cor_5}) must have been true for either $i-1$ or $i-2$ and (\ref{tw_cor_6}) for the other, while for vertices $v_j$ of high state $s_j$ it suffices that (\ref{tw_cor_6}) must have been true for both.
\end{itemize}

 For the algorithm's complexity, there are $k\cdot d^{\tw}$ entries for each table $D_i$ of any node $i\in I$, with $|I|=O(\tw\cdot|V|)$ for nice tree decomposition $(\mathcal{X},T=(I,F))$, while any entry can be computed in time $O(1)$ for leaf and introduce nodes, $O(d)$ for forget nodes, while the state changes can be computed in $O(k\cdot\tw\cdot d^\tw)$ time, with each entry of the transformed table $D_i^*$ computed in $O(k)$ time.  
\end{proof}

\section{Vertex Cover, Feedback Vertex Set: W[1]-Hardness}\label{sec_vc}

In this section we show that the edge-weighted variant of the \dS\ problem parameterized by $\vc+k$ is W[1]-hard via a reduction from \textsc{$k$-Multicolored independent Set}.

\paragraph{Construction:} Given an instance $[G=(V,E),k]$ of \textsc{$k$-Multicolored independent Set}, we construct an instance $[G'=(V',E'),k']$ of edge-weighted \dS\, where $d=6n$. First, for every color class $V_i\subseteq V$ we create a set $P_i\subseteq V'$ of $n$ vertices $p_l^i,\forall l\in[1,n],\forall i\in[1,k]$ (that directly correspond to the vertices of $V_i$). Next, for each $i\in[1,k]$ we make a pair of vertices $a_i,b_i$, connecting $a_i$ to each vertex $p_l^i$ by an edge of weight $n+l$, while $b_i$ is connected to each vertex $p_l^i$ by an edge of weight $2n-l$. Next, for every non-edge $e\in \bar{E}$ (i.e.\ $\bar{E}$ contains all pairs of vertices from $V$ that are not connected by an edge from $E$) between two vertices from different $V_{i_1},V_{i_2}$ (with $i_1\not=i_2$), we make a vertex $u_e$ that we connect to vertices $a_{i_1},b_{i_1}$ and $a_{i_2},b_{i_2}$. We set the weights of these edges as follows: suppose that $e$ is a non-edge between the $j_1$-th vertex of $V_{i_1}$ and the $j_2$-th vertex of $V_{i_2}$. We then set $w(u_e,a_{i_1})=5n-j_1$, $w(u_e,b_{i_1})=4n+j_1$ and $w(u_e,a_{i_2})=5n-j_2$, $w(u_e,b_{i_2})=4n+j_2$. Next, for every pair of $i_1,i_2$ we make two vertices $g_{i_1,i_2}$, $g'_{i_1,i_2}$. We connect $g_{i_1,i_2}$ to all vertices $u_e$ that correspond to non-edges $e$ between vertices of the same pair $V_{i_1},V_{i_2}$ by edges of weight $(6n-1)/2$ and also $g_{i_1,i_2}$ to $g'_{i_1,i_2}$ by an edge of weight $(6n+1)/2$. In this way, a $k$-multicolored independent set in $G$ corresponds to a $6n$-scattered set in $G'$ of size $k^2$. This concludes the construction of $G'$, with Figure \ref{fig:Whard_VC} providing an illustration.

   \begin{figure}[htbp]
  \centerline{\includegraphics[width=110mm]{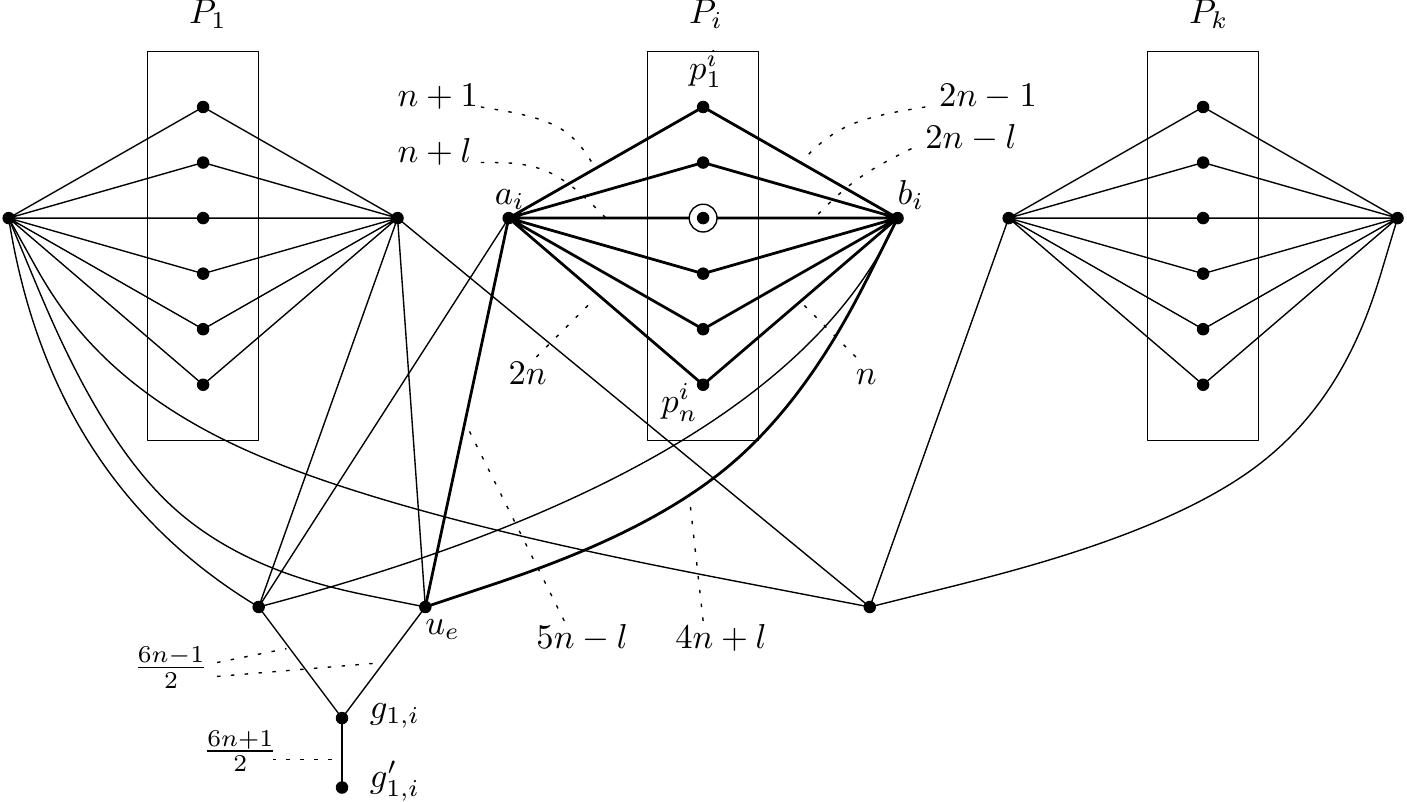}}
  \caption{A general picture of graph $G'$, where the circled vertex is $p_l^i$ and dotted lines match weights to edges.}
  \label{fig:Whard_VC} \end{figure}
  
\begin{lemma}\label{Whard-FWD}
  If $G$ has a $k$-multicolored independent set, then $G'$ has a $6n$-scattered set of size $k+2{k \choose 2}=k+k(k-1)=k^2$.
\end{lemma}
\begin{proof}
 Let $I\subseteq V$ be a multicolored independent set in $G$ of size $k$ and $v_{l_i}^i$ denote the vertex selected from each $V_i$, or $I\coloneqq\{v_{l_1}^1,\dots,v_{l_k}^k\}$. Let $S\subseteq V'$ include the set of vertices $p_{l_i}^i$ in $G'$ that correspond to each $v_{l_i}^i$. For any pair $i,j\in[1,k]$ of indices with $i\not=j$, let $u_e$ be the vertex corresponding to the non-edge between vertices $v_{l_i}^i,v_{l_j}^j\in I$. All these $u_e$ vertices exist, as $I$ is a $k$-multicolored independent set. We include all these $u_e$ vertices in $S$ and also every $g'_{i,j}$ that is connected to $g_{i,j}$ that each such $u_e$ is connected to. Now $S$ is of size $k+2{k \choose 2}$ and we claim it is a $6n$-scattered set: all selected vertices $p_{l_i}^i$ are at distance $n+l_i+5n-l_i=6n$ via $a_i$ and distance $2n-l_i+4n+l_i=6n$ via $b_i$ from any selected vertex $u_e$ that corresponds to a non-edge ``adjacent'' to their corresponding $v_{l_i}^i$, while every selected vertex $u_e$ corresponding to a non-edge between two vertices of groups $V_{i_1},V_{i_2}$ is at distance $(6n-1)/2+(6n+1)/2=6n$ from every selected vertex $g'_{i_1,i_2}$ that is connected to $g_{i_1,i_2}$ that connects all such $u_e$ vertices between these groups.  
\end{proof}

\begin{lemma}\label{Whard-BWD}
  If $G'$ has a $6n$-scattered set of size $k+2{k\choose 2}=k+k(k-1)=k^2$, then $G$ has a $k$-multicolored independent set.
\end{lemma}
\begin{proof}
 Let $S\subseteq V'$ be the $6n$-scattered set, with $|S|=k+2{k\choose 2}$. As the distance via $g_{i,j}$ between any two vertices $u_e,u_h$ corresponding to non-edges between vertices of the same groups $V_i,V_j$ is $6n-1$, set $S$ can contain at most one such vertex for every such pair of groups, their number being ${k\choose 2}$. Since the size of $S$ is $k+2{k\choose 2}$, the set must also contain one other vertex per group, the only choices available being vertices $g'_{i,j}$ at distance $(6n-1)/2+(6n+1)/2=6n$ from any such $u_e$, leaving the $k$ choices for at most one vertex from each $P_i$, as the distance between any pair $p_{l_1}^i,p_{l_2}^i$ is $2n+l_1+l_2<6n$ via $a_i$ and $4n-l_1-l_2<6n$ via $b_i$. Now, let $u_e\in S$ be a selected vertex corresponding to a non-edge between vertices $v_{l_i}^i,v_{l_j}^j$ from groups $V_i,V_j$ and $p_{o_i}^i,p_{o_j}^j\in S$ be the vertices selected from $P_i,P_j$. Vertex $u_e$ is at distance $5n-l_i+n+o_i=6n+o_i-l_i$ via $a_i$ and $4n+l_i+2n-o_i=6n+l_i-o_i$ via $b_i$ from $p_{o_i}^i\in P_i$, while at distance $5n-l_j+n+o_j=6n+o_j-l_j$ via $a_j$ and $4n+l_j+2n-o_j=6n+l_j-o_j$ via $b_j$ from $p_{o_j}^j\in P_j$. It is not hard to see that if $o_i\not=l_i$ then $u_e$ and $p_{l_i}^i$ cannot be together in $S$, while also if $o_j\not=l_j$ then $u_e$ and $p_{l_j}^j$ cannot be together in $S$. Thus, there must be no edge between every pair of vertices $v_{l_i}^i,v_{l_j}^j$ that correspond to $p_{l_i}^i,p_{l_j}^j\in S$, meaning the set $I$ that includes all such $v_{l_i}^i$ is a $k$-multicolored independent set.  
\end{proof}

\begin{theorem}\label{vc_Whard_THM}
 The edge-weighted \dS\ problem is W[1]-hard parameterized by $\vc+k$. Furthermore, if there is an algorithm for edge-weighted \dS\ running in time $n^{o(\sqrt{\vc}+\sqrt{k})}$ then the ETH is false.
\end{theorem}
\begin{proof}
 Observe that the set $Q\subset V'$ that includes all vertices $a_i,b_i,\forall i\in[1,k]$ and all vertices $g_{i,j},\forall i\not=j\in[1,k]$ is a vertex cover of $G'$, as all edges have exactly one endpoint in $Q$. This means $\vc(G')\le2k+{k\choose 2}=O(k^2)$. In addition, the relationship between the sizes of the solutions of \dS\ and \textsc{$k$-Multicolored Independent Set} is $k'=|S|=k+2{k\choose 2}=O(k^2)$. Thus, the construction along with lemmas \ref{Whard-FWD} and \ref{Whard-BWD}, indeed imply the statement.  
\end{proof}

Using essentially the same reduction (with minor modifications) we also obtain similar hardness results for unweighted \dS\ parameterized by $\fvs$:

\begin{corollary}\label{fvs_cor}
 The unweighted \dS\ problem is W[1]-hard parameterized by $\fvs+k$. Furthermore, if there is an algorithm for unweighted \dS\ running in time $n^{o(\fvs+\sqrt{k})}$ then the ETH is false.
\end{corollary}
\begin{proof}
 The modifications to the above construction that we require are the following: each edge $e$ of weight $w(e)$ is substituted by a path of length $w(e)$, apart from the edge between every $g_{i_1,i_2}$ to every $g'_{i_1,i_2}$ that is now a path of length $d-1=6n-1$ and all edges between every $g_{i_1,i_2}$ to all adjacent $u_e$ that correspond to non-edges between vertices of pair $V_{i_1},V_{i_2}$ that are now only a single edge. In this way, Lemma \ref{Whard-FWD} goes through unchanged, while for Lemma \ref{Whard-BWD}, it suffices to observe that no two vertices anywhere on the paths between some $g_{i_1,i_2}$ and some $a_i,b_i$ could be selected instead of the intended selection of $g'_{i_1,i_2}$ and some $u_e$ that matches the selections from $V_{i_1},V_{i_2}$, as the distance between any two vertices between $g_{i_1,i_2}$ and some $a_i,b_i$ is always $<2d=12n$, while if the selected vertices are not exactly some $g'_{i_1,i_2}$ and (the correct) $u_e$, then the minimum distance between these selections and the closest selection from $V_{i_1},V_{i_2}$ will be less than $d$.
 
 It is not hard to see that the set $Q$ containing all $a_i,b_i$, $\forall i\in[1,k]$ is a feedback vertex set of $G'$, as removal of all these vertices results in an acyclic graph, hence $\fvs(G')\le O(k)$.  
\end{proof}

\section{Vertex Cover: FPT Algorithm}\label{sec_vc_algo}

We next show that unweighted \dS\ admits an FPT algorithm parameterized by $\vc$, in contrast to its
weighted version (Theorem \ref{vc_Whard_THM}). Given graph $G$ along with a vertex cover $C$ of $G$ and $d\ge3$, our
algorithm first defines an instance of \textsc{Partial Set Packing}, where elements may be \emph{partially} included in some sets and then solves the problem by dynamic programming. In this variant, any element has a \emph{coefficient} of inclusion in each set and a collection of sets is a solution if there is no pair of sets for which the sum of any element's coefficients is $>1$.

We make a set for each vertex and an element for each vertex of $C$. Our aim is to identify two vertices (sets) as incompatible selections if there is some third ``middle'' vertex from $C$ (elements), whose sum of distances to the other two is $<d$, based on the observation that for any vertex not belonging to the $d$-scattered set, only one selection can be at distance $<d/2$, yet any number of selections can be at distance $\ge d/2$ (consider a star as an example).

These coefficients of inclusion are then used to assign vertices of $C$ to their closest possible selections, with complete inclusion (i.e.\ coefficient equal to 1) implying the distance is $<d/2$ and no inclusion (equal to 0) that it is $>d/2$. For the middle vertices, depending on the parity of $d$ (and causing the difference in running times), we require either one (i.e.\ $1/2$) or two ($1/3$ and $2/3$) extra coefficients to be able to determine the exact position of a possible middle vertex from $C$ (element) on the path between two potential selections (sets). If the sum of coefficients is $\le1$, the vertex from $C$ is either a middle vertex on the path between the two selections or at distance $<d/2$ from only one of them. On the other hand, if the sum of coefficients is $>1$, then the sum of distances from the vertex to the two selections is $<d$ and the incompatibility of the sets implies the corresponding vertices cannot both belong in the $d$-scattered set.

\begin{theorem}\label{vc_Algo_THM}
 Given graph $G$, along with $d>2$ and a vertex cover of size $\vc$ of $G$, there exists an algorithm solving the unweighted \dS\ problem in $O^*(3^{\vc})$ time for even $d$ and $O^*(4^{\vc})$ time for odd $d$.
\end{theorem}
\begin{proof}
Let $C$ be the given vertex cover of $G$ and $I=V\setminus C$ be the remaining independent set. Let also $Y\subseteq I$ be the subset of vertices from $I$ with a unique neighborhood in $C$, i.e.\ for two vertices $u,v\in I$ with $N(u)=N(v)$, set $Y$ only contains one of them. Observe that the size of $Y$ is thus exponentially bounded by the size of $C$: $|Y|\le2^{|C|}$.

For an instance of our \textsc{Partial Set Packing} variant, let $\mathcal{U}=\{u_1,\dots,u_n\}$ be the universe of elements and $\mathcal{S}=\{S_1,\dots,S_m\}$ be the set family. Further, for even $d$, we introduce a weight function $w(u_i,S_j):\mathcal{U}\times\mathcal{S}\mapsto\{0,1/2,1\}$, giving the \emph{coefficient} of element $u_i$ for inclusion in set $S_j$, where 0 implies the element is not included in the set, 1/2 implies \emph{partial} and 1 \emph{complete} inclusion. For odd $d$, the  weight function $w(u_i,S_j):\mathcal{U}\times\mathcal{S}\mapsto\{0,1/3,2/3,1\}$ allows more values for partial inclusion. In our solutions to this variant we will allow any number of sets to partially include any element, yet if any set in the solution completely includes some element, then no other set that includes the same element either partially, or completely, can also be part of the same solution, i.e.\ a collection of subsets $S\subseteq\mathcal{S}$ will be a solution, if for every element $u_i\in\mathcal{U}$, the sum for any two pairs is at most 1: $\max_{S_a,S_b\in S}\{w(u_i,S_a)+w(u_i,S_b)\}\le1$.

We then define our \textsc{Partial Set Packing} instance as follows: we make an element $u_i\in\mathcal{U}$ for every vertex of $C$ and a set $S_j\in\mathcal{S}$ for every vertex of $C\cup Y$. We thus have $|C|$ elements and $|C|+|Y|\le|C|+2^{|C|}$ sets. For even $d$, an element $u_i$ with corresponding vertex $v\in C$ is included in some set $S_j$ with corresponding vertex $z$ completely (or $w(u_i,S_j)=1$) if $d(v,z)<d/2$, while an element $u_i$ with corresponding vertex $v\in C$ is included in some set $S_j$ with corresponding vertex $z$ partially (or $w(u_i,S_j)=1/2$), if $d(v,z)=d/2$. For odd $d$, an element $u_i$ with corresponding vertex $v\in C$ is included in some set $S_j$ with corresponding vertex $z$ completely (or $w(u_i,S_j)=1$) if $d(v,z)<\lfloor d/2\rfloor$, 2/3-partially ($w(u_i,S_j)=2/3$) if $d(v,z)=\lfloor d/2\rfloor$ and 1/3-partially ($w(u_i,S_j)=1/3$) if $d(v,z)=\lceil d/2\rceil$.  

In the classic dynamic programming procedure for \textsc{Set Packing} we store a table $OPT[U,j]$ that contains, for every subset of elements $U\subseteq\mathcal{U}$ and $j\in[1,m]$ the maximum number of subsets that can be selected from $\{S_1,\dots,S_j\}$, such that no element of $U$ is included in any of them. The dynamic programming procedure then first computes for $j=1$: $OPT[U,1]\coloneqq 1$, if $U\cap S_j=\emptyset$ and 0 otherwise, while for $j=2,\dots,m$ it is: $OPT[U,j+1]\coloneqq\max\{OPT[U,j],OPT[U\cup S_{j+1},j]+1\}$ if $S_{j+1}\cap U=\emptyset$ and only $OPT[U,j+1]\coloneqq\ OPT[U,j]$ otherwise.

We will create a similar table $OPT[U,j]$ for every $j\in[1,m]$ and every $U=\{(u_i\in \mathcal{U},w(u_i,U))\}$ (of the possible $3^{|\mathcal{U}|}$ or $4^{|\mathcal{U}|}$), storing the maximum number of sets that can be selected from $\{S_1,\dots,S_j\}$ to form a partial solution $S'\subseteq\{S_1,\dots,S_j\}$, so that for any element $u_i$ it is $\max_{S_l\in S'}\{w(u_i,S_l)+w(u_i,U)\}\le1$. Letting the union operator $A\cup B$ transfer maximum inclusion, i.e.\ $w(u_i,A\cup B)=\max\{w(u_i,A),w(u_i,B)\}$, and substituting the check for $U\cap S_j=\emptyset$ by $\forall u_i\in U\cup S_j: w(u_i,U)+w(u_i,S_j)\le1$ in the above procedure, we can solve the \textsc{Partial Set Packing} instance in $O(mn4^n)$ time (and only $O(mn3^n)$ for even $d$).

Given a solution $S\subseteq\mathcal{S}$ to our \textsc{Partial Set Packing} instance, we will show that it corresponds to a solution for the original instance of \dS. First observe that on any shortest path $v_0,\dots,v_d$ between vertices $v_0,v_d$, we know that any vertex $v_i$ will either be included in $C$, or both its neighbors $v_{i-1},v_{i+1}$ on the path will be included instead, as otherwise both edges adjacent to $v_i$ are not covered by $C$.

Consider first the case where $d$ is even. On the shortest path between two vertices $v_0,v_d$ that are at distance $d$ from each other there will be one (\emph{middle}) vertex $v_{d/2}$ at distance $d/2$ from both and if $v_{d/2}\in C$ then the corresponding element will be partially included in both sets corresponding to $v_0,v_d$, while if $v_{d/2}\notin C$, each of the elements corresponding to its neighbors $v_{d/2-1},v_{d/2+1}$ on the path will be completely included in one set each and thus both sets can be used in solution $S$. For two vertices $v_0,v_{d-1}$ at distance $d-1$ from each other, there will be one vertex $v_{d/2}$ at distance $d/2$ from $v_0$ and $d/2-1$ from $v_{d-1}$ and also one vertex $v_{d/2-1}$ at distance $d/2-1$ from $v_0$ and $d/2$ from $v_{d-1}$. If $v_{d/2}\in C$, then its corresponding element is included partially by 1/2 in the set corresponding to vertex $v_0$ and completely by 1 in the set corresponding to vertex $v_{d-1}$. Otherwise, if $v_{d/2-1}\in C$, then its corresponding element is included completely by 1 in the set corresponding to vertex $v_0$ and partially by 1/2 in the set corresponding to vertex $v_{d-1}$. Thus in both cases these two sets cannot be included together in $S$. The argument also holds if the distance between the two vertices is smaller than $d-1$.

Next, if $d$ is odd, on the shortest path between two vertices $v_0,v_d$ that are at distance $d$ from each other there will be two middle vertices $v_{\lfloor d/2\rfloor},v_{\lceil d/2\rceil}$ at distances $\lfloor d/2\rfloor$ and $\lceil d/2\rceil$ from each. Now, vertex $v_{\lfloor d/2\rfloor}$ will be at distance $\lfloor d/2\rfloor$ from $v_0$ and distance $\lceil d/2\rceil$ from $v_d$ and if $v_{\lfloor d/2\rfloor}\in C$ its element will be included by 2/3 in the set corresponding to $v_0$ and by 1/3 in the set corresponding to $v_d$. Similarly, if $v_{\lceil d/2\rceil}\in C$, its element will be included by 1/3 in the set corresponding to $v_0$ and by $2/3$ in the set corresponding to $v_d$. Thus in both cases the two sets can be included together in $S$. For two vertices $v_0,v_{d-1}$ at distance $d-1$ from each other, if vertex $v_{\lfloor d/2\rfloor}\in C$, then its element will be included by 2/3 in both sets corresponding to $v_0,v_{d-1}$, while if $v_{\lfloor d/2\rfloor}\notin C$, then we have that both its neighbors $v_{\lfloor d/2\rfloor-1}$, $v_{\lfloor d/2\rfloor+1}\in C$. Now, the element corresponding to $v_{\lfloor d/2\rfloor-1}$ will be completely included by 1 in the set corresponding to $v_0$ and partially by 1/3 in the set corresponding to $v_{d-1}$, while the element corresponding to $v_{\lfloor d/2\rfloor+1}$ will likewise be included partially by 1/3 in the set corresponding to $v_0$ and completely by 1 in the set corresponding to $v_{d-1}$. Thus in both cases, these two sets cannot be included together in $S$, while the argument also holds if the distance between the vertices is smaller than $d-1$.

As the number of sets in our \textsc{Partial Set Packing} instance is $m=|C|+|Y|\le|V|$ and the number of elements is $n=|C|=\vc$, the total running time of our algorithm is bounded by $O^*(4^{\vc})$ for odd $d$ and $O^*(3^{\vc})$ for even $d$.  
\end{proof}

\section{Tree-depth: Tight ETH Lower Bound}\label{sec_td}
 
In this section we consider the unweighted version of the \dS\ problem parameterized by $\td$. We first show the existence of an FPT algorithm of running time $O^*(2^{O(\td^2)})$ and then a tight ETH-based lower bound. We begin with a simple upper bound argument, making use of the following fact on tree-depth, while the algorithm then follows from the dynamic programming procedure of Theorem \ref{tw_DP_THM} and the relationship between $d,\td$ and $\tw$:

\begin{lemma}\label{lem:td-diam} For any graph $G=(V,E)$ we have $D(G)\le 2^{\td+1}-2$, where $D(G)$ denotes the graph's diameter.
\end{lemma}
\begin{proof}
 We use the following equivalent inductive definition of tree-depth:
$\td(K_1)=0$ while for any other graph $G=(V,E)$ we set $\td(G) = 1+ \min_{u\in
 V} \td(G\setminus u)$ if $G$ is connected, and $\td(G) = \max_C \td(G[C])$ if
 $G$ is disconnected, where the maximum ranges over all connected components of
 $G$.
 
 We prove the claim by induction. The inequality is true for $K_1$, whose
 diameter is $0$. For the inductive step, the interesting case is when $G=(V,E)$
 is connected, since otherwise we can assume that the claim has been shown for
 each connected component and we are done. Let $u\in V$ be such that $\td(G) =
 1+\td(G\setminus u)$. Consider two vertices $v_1,v_2\in V\setminus\{u\}$ which
 are at maximum distance in $G$. If $v_1,v_2$ are in the same connected
 component of $G':=G\setminus u$, then $d_{G}(v_1,v_2) \le d_{G'}(v_1,v_2) \le
 D(G') \le 2^{\td(G')+1}-2 \le 2^{\td(G)+1} -2$, where we have used the
 inductive hypothesis on $G'$. So, suppose that $v_1,v_2$ are in different
 connected components of $G'$. It must be the case that $u$ has a neighbor in
 the component of $v_1$ (call it $v_1'$) and in the component of $v_2$ (call it
 $v_2'$), because $G$ is connected. We have $d_G(v_1,v_2) \le d_G(v_1,v_1') + 2
 + d_G(v_2,v_2') \le d_{G'}(v_1,v_1') + 2 + d_{G'}(v_2,v_2') \le 2D(G') +2 \le
 2\cdot 2^{\td(G')+1} - 2 = 2^{\td(G)+1}-2$.  
 \end{proof}

\begin{theorem}\label{thm:td-alg}
 Unweighted \dS\ can be solved in time $O^*(2^{O(\td^2)})$.
 \end{theorem}
 \begin{proof} The main observation is that we can assume that $d\le D(G)$,
 because otherwise the problem is trivial. Hence, by Lemma \ref{lem:td-diam} we
 have $d\le 2^{\td+1}$. We can now rely on Lemma \ref{lem:relations} to get
 $\tw\le\td$, and the algorithm of Theorem \ref{tw_DP_THM} which runs in time
 $O^*(d^\tw)$ gives the desired running time.  
 \end{proof}
 
 Next we show a lower bound matching Theorem \ref{thm:td-alg}, based on the ETH, using a reduction from \textsc{3-SAT} and a construction similar to the one used in Section \ref{sec_vc}.
 
 \paragraph{Construction:}  Given an instance $\phi$ of \textsc{3-SAT} on $n$ variables and $m$ clauses, where we can assume that $m=O(n)$ (by the Sparsification Lemma, see \cite{ImpagliazzoPZ01}), we will create an instance $[G=(V,E)]$ of the unweighted \dS\ problem where $d=6\cdot c^{\sqrt{n}}$ for an appropriate constant $c$ (to simplify notation, we assume without loss of generality that $\sqrt{n}$ is an integer).  We first group the clauses of $\phi$ into $\sqrt{n}$ equal-sized groups $F_1,\dots,F_{\sqrt{n}}$ and as a result, each group involves $O(\sqrt{n})$ variables, with $2^{O(\sqrt{n})}$ possible assignments to the variables of each group. We select $c$ appropriately so that each group $F_i$ has at most $c^{\sqrt{n}}$ possible partial assignments $\phi_j^i$ for the variables of clauses in $F_i$.
 
 We then create for each $i\in\{1,\dots,\sqrt{n}\}$, a set $P_i$ of at most $c^{\sqrt{n}}$ vertices $p_1^i,\dots,p_{c^{\sqrt{n}}}$, such that each vertex of $P_i$ represents a partial assignment to the variables of $F_i$ that satisfies all clauses of $F_i$. We then create for each $i\in\{1,\dots,\sqrt{n}\}$ a pair of vertices $a_i,b_i$ and we connect $a_i$ to each vertex $p_l^i$ by a path of length $c^{\sqrt{n}}+l$, while $b_i$ is connected to each vertex $p_l^i$ by a path of length $2\cdot c^{\sqrt{n}}-l$. Now each $P_i$ contains all $a_i,b_i$ and $p_l^i,i\in\{1,\dots,c^{\sqrt{n}}\}$.
 
 Finally, for every two \emph{non-conflicting} partial assignments $\phi_l^i,\phi_o^j$, with $l,o\in[1,c^{\sqrt{n}}]$ and $i,j\in[1,\sqrt{n}]$, i.e.\ two partial assignments that do not assign conflicting values to any variable, we create a vertex $u_{l,o}^{i,j}$ that we connect to vertices $a_i,b_i$ and $a_j,b_j$: if $p_l^i\in P_i$ is the vertex corresponding to $\phi_l^i$ and $p_o^j\in P_j$ is the vertex corresponding to $\phi_o^j$, then vertex $u_{l,o}^{i,j}$ is connected to $a_i$ by a path of length $5\cdot c^{\sqrt{n}}-l$ and to $b_i$ by a path of length $4\cdot c^{\sqrt{n}}+l$, as well as to $a_j$ by a path of length $5\cdot c^{\sqrt{n}}-o$ and to $b_j$ by a path of length $4\cdot c^{\sqrt{n}}+o$. Next, for every pair $i,j$ we make two vertices $g_{i,j},g'_{i,j}$. We connect $g_{i,j}$ to all vertices $u_{l,o}^{i,j}$ (for any $l,o$) by a single edge and also $g_{i,j}$ to $g'_{i,j}$ by a path of length $6\cdot c^{\sqrt{n}}-1$. This concludes our construction and Figure \ref{fig:Whard_td2} provides an illustration.
 
\begin{figure}[htbp]
\centerline{\includegraphics[width=120mm]{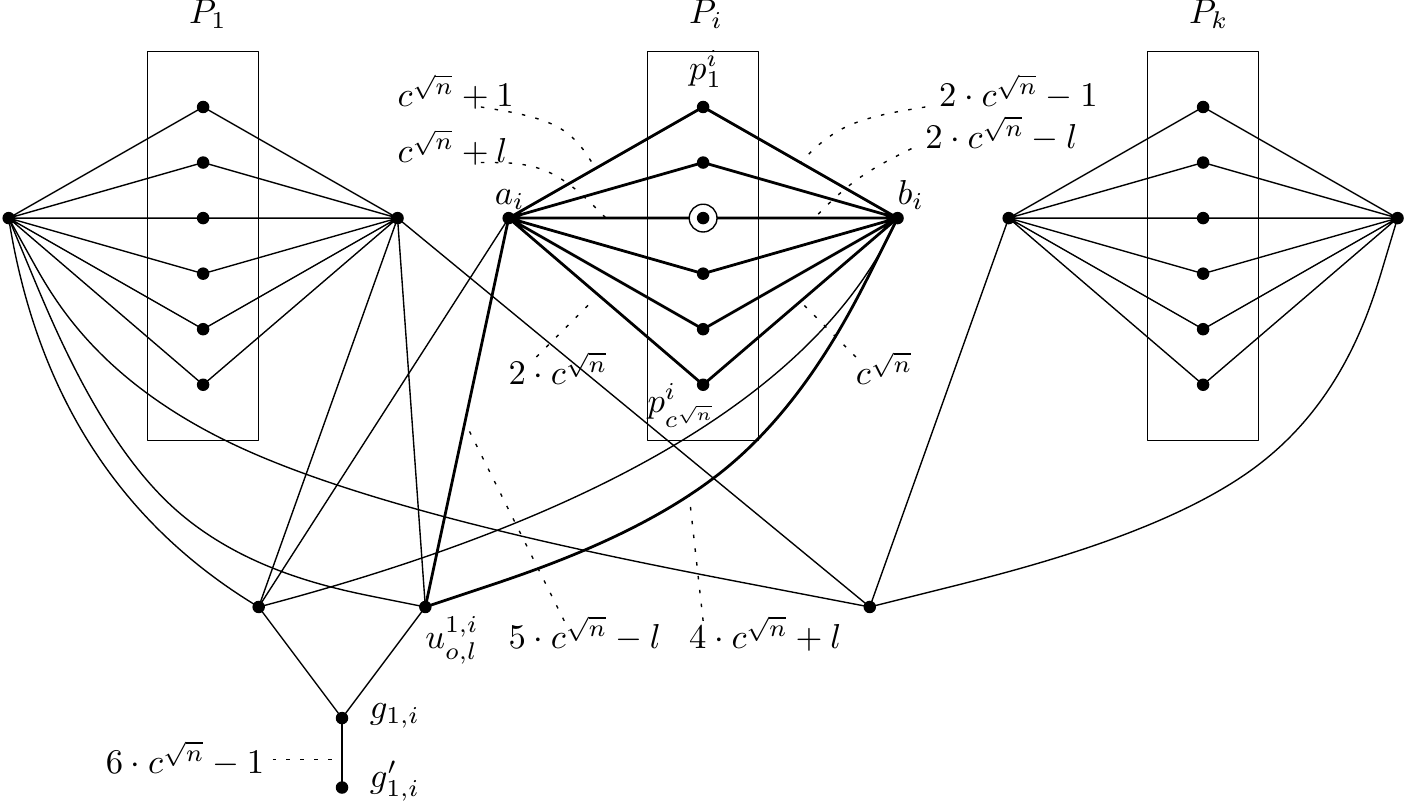}}
\caption{A general picture of graph $G$, where straight lines imply paths of length equal to the number indicated by dotted lines, while the circled vertex is $p_l^i$.}
\label{fig:Whard_td2} \end{figure}

\begin{lemma}\label{td_ETH_FWD}
   If $\phi$ has a satisfying assignment, then there exists a $6\cdot c^{\sqrt{n}}$-scattered set in $G$ of size $\sqrt{n}+2{\sqrt{n} \choose 2}=n$.
  \end{lemma}
 \begin{proof}
  Consider the satisfying assignment for $\phi$ and let $\phi_{l_i}^i$, with $l_i\in[1,c^{\sqrt{n}}]$ and $i\in[1,\sqrt{n}]$, be the restriction of that assignment for all variables appearing in clauses of group $F_i$. We claim the set $K$, consisting of all vertices $p_{l_i}^i$ corresponding to $\phi_{l_i}^i$, all vertices $g'_{i,j}$ and all $u_{l,o}^{i,j}$ vertices for which we have selected $p_{l}^i$ and $p_{o}^j$ (all these vertices exist, as the corresponding partial assignments are non-conflicting), is a $d$-scattered set for $G$ of size $|K|=\sqrt{n}+2{\sqrt{n} \choose 2}=n$: all selected vertices $p_{l_i}^i$ are at distance $c^{\sqrt{n}}+l_i+5\cdot c^{\sqrt{n}}-l_i=6\cdot c^{\sqrt{n}}$ via $a_i$ and distance $2\cdot c^{\sqrt{n}}-l_i+4\cdot c^{\sqrt{n}}+l_i=6\cdot c^{\sqrt{n}}$ via $b_i$ from any selected vertex $u_{l_i,l_j}^{i,j}$, while every selected $u_{l_i,l_j}^{i,j}$ is at distance $6\cdot c^{\sqrt{n}}-1+1=6\cdot c^{\sqrt{n}}$ from every selected $g'_{i,j}$.  
 \end{proof}
 
 \begin{lemma}\label{td_ETH_BWD}
   If there exists a $6\cdot c^{\sqrt{n}}$-scattered set in $G$ of size $\sqrt{n}+2{\sqrt{n} \choose 2}=n$, then $\phi$ has a satisfying assignment.
  \end{lemma}
 \begin{proof}
  Let $S\subseteq V$ be the $6\cdot c^{\sqrt{n}}$-scattered set in $G$, with $|S|=n$. For every pair $i,j\in[1,\sqrt{n}]$, set $S$ cannot contain more than one vertex from the paths between $g_{i,j}$ and $a_i,b_i,a_j,b_j$, as the distance between any pair of such vertices is always $<2\cdot 6\cdot c^{\sqrt{n}}$ (due to the single edges between $g_{i,j}$ and any $u_{l_i,l_j}^{i,j}$). Likewise, set $S$ cannot contain more than two vertices from the paths between $g'_{i,j}$ and $a_i,b_i,a_j,b_j$, as the maximum sum of distances between any three such vertices is $<3\cdot 6\cdot c^{\sqrt{n}}$. Since $|S|=\sqrt{n}+2{\sqrt{n} \choose 2}$, set $S$ must also contain $\sqrt{n}$ other vertices and due to the distance between any pair of vertices $p_{l}^i,p_o^i$ from the same group $P_i$ being $<4\cdot \sqrt{n}$, there must be one selection from each group $P_i$. Furthermore, for two such selections $p_{l_i}^i,p_{l_j}^j$, the only option for the other two selections (for this pair of groups $i,j$) is to select vertices $g'_{i,j}$ and $u_{l_i,l_j}^{i,j}$, since the distances from $p_{l_i}^i,p_{l_j}^j$ to $u_{l_i,l_j}^{i,j}$ (through $a_i,b_i,a_j,b_j$) will only be equal to $6\cdot c^{\sqrt{n}}$ if these selections (and indices) match, with the only remaining option at distance $6\cdot c^{\sqrt{n}}$ (for any choice of $u_{l_i,l_j}^{i,j}$) being vertex $g'_{i,j}$.  
 \end{proof}
 
 \begin{lemma}\label{td_ETH_bound}
   The tree-depth of $G$ is $2\sqrt{n}+\lceil\log(6\cdot c^{\sqrt{n}})\rceil+1=O(\sqrt{n})$.
  \end{lemma}
 \begin{proof}
  We again employ the alternative definition of tree-depth used in Lemma \ref{lem:td-diam}. Consider graph $G$ after removal of all vertices $a_i,b_i,\forall i\in[1,\sqrt{n}]$. The graph now consists of $\sqrt{n}\cdot c^{\sqrt{n}}$ paths of length $<3\cdot c^{\sqrt{n}}$ through each vertex in $P_i$ and ${\sqrt{n}\choose 2}$ trees, considered rooted at each vertex $g_{i,j}$. The maximum distance in each such tree between a leaf and its root is $6\cdot c^{\sqrt{n}}-1$ (for vertex $g'_{i,j}$) and the claim then follows, as paths of length $n$ have tree-depth exactly $\lceil\log(n+1)\rceil$ (this can be shown by repeatedly removing the middle vertex of each path). By the definition of tree-depth, after removal of $2\sqrt{n}$ vertices from $G$, the maximum tree-depth of each resulting disconnected component is $\lceil\log(6\cdot c^{\sqrt{n}})\rceil=\lceil\sqrt{n}\cdot\log(c)+\log(6)\rceil$.  
 \end{proof}

 \begin{theorem}\label{td_ETH_thm}
  If unweighted \dS\ can be solved in $2^{o(\td^2)}\cdot n^{O(1)}$ time, then \textsc{3-SAT} can be solved in $2^{o(n)}$ time.
 \end{theorem}
 \begin{proof}
  Suppose there is an algorithm for \dS\ with running time $2^{o(\td^2)}$. Given an instance $\phi$ of \textsc{3-SAT}, we use the above construction to create an instance $[G=(V,E)]$ of \dS with $d=6\cdot c^{\sqrt{n}}$, in time $O(\sqrt{n}\cdot c^{\sqrt{n}}+c^{2\sqrt{n}})$. As, by Lemma \ref{td_ETH_bound}, we have $\td(G)\le O(\sqrt{n})$, using the supposed algorithm for \dS\ we can decide whether $\phi$ has a satisfying assignment in time $2^{o(\td^2)}\cdot n^{O(1)}=2^{o(n)}$.  
 \end{proof}

\section{Treewidth Revisited: FPT-AS}\label{sec_tw_approx} Here we
present an FPT approximation \emph{scheme} (FPT-AS) for \dS\ parameterized by
$\tw$. Given as input an edge-weighted graph $G=(V,E)$, $k\in\mathbb{N}^+$,
$d\ge2$ and an arbitrarily small error parameter $\epsilon>0$, our algorithm is
able to return a set $K$, such that any $v,u\in K$ are
at distance $d(v,u)\ge\frac{d}{1+\epsilon}$, in time
$O^*((\tw/\epsilon)^{O(\tw)})$, if $G$ has a $d$-scattered set of size $|K|$.

Our algorithm makes use of a technique introduced in \cite{Lampis14} (see also
\cite{AngelBEL16,KatsikarelisLP17}) for approximating problems that
are W-hard by treewidth. If the hardness of the problem arises from the need of the dynamic programming table to store $\tw$ large numbers (in our
case, the distances of the vertices in the bag from the closest selection), we
can significantly speed up the algorithm by replacing all values by the closest
integer power of $(1+\delta)$, for some appropriately chosen $\delta$, thus
reducing the table size from $d^\tw$ to $(\log_{(1+\delta)}d)^\tw$.
Of course, the calculations may result in values that are not integer powers of
$(1+\delta)$ that will thus have to be ``rounded'' to maintain the table size.
This might introduce the accumulation of rounding errors, yet we are able to show
that the error on any rounded value can be bounded by a function of the height
of its corresponding bag and then make use of a theorem from \cite{BodlaenderH98} stating
that any tree decomposition can be balanced so that its width remains almost
unchanged and its total height becomes $O(\log n)$.

The rounding technique as applied in \cite{Lampis14} employs randomization and
an extensive analysis to procure the bounds on the propagation of error, while
we only require a deterministic adaptation of the rounding process without
making use of the advanced machinery there introduced, as for our particular
case, the bound on the rounding error can be straightforwardly
obtained.  The main tool we require is the following definition of an
addition-rounding operation, denoted by $\oplus$:
for two non-negative numbers $x_1,x_2$, we
define $x_1\oplus x_2\coloneqq0$, if $x_1=x_2=0$. Otherwise, we set $x_1\oplus
x_2\coloneqq(1+\delta)^{\lfloor\log_{(1+\delta)}(x_1+x_2)\rfloor}$.

The integers we would like to approximately store are the states $s_j\in[1,d-1]$, representing the distance of a vertex $v_j$ in bag $X_i$ of the tree decomposition to the closest selection in the $d$-scattered set $K$, during computation of the dynamic programming algorithm.
Let $\Sigma_{\delta}\coloneqq\{0\}\cup\{(1+\delta)^l|l\in\mathbb{N}\}$. Intuitively, $\Sigma_\delta$ is the set of rounded states that our modified algorithm may use. Of course, $\Sigma_\delta$ as defined is infinite, but we will only consider the set of values that are at most $d$, denoted by $\Sigma_\delta^d$. In this way,
the size of $\Sigma_\delta^d$ is reduced to $\log_{(1+\delta)}(d)$, that
for $\delta=\frac{\epsilon}{O(\log n)}$, gives $|\Sigma_\delta^d|=O(\log(d)\log(n)/\epsilon)$ and we then rely on the well-known win-win parameterized argument given in Section \ref{sec_prelim} (Lemma~\ref{lem:fpt-logn}) to get a running time of $O^*((\tw/\epsilon)^{O(\tw)})$.

\paragraph{Modifications:} Our approximation algorithm will be a modification of the exact dynamic programming for \dS, given in Section \ref{sec_tw_dp}. For the approximation algorithm, we will make use of an adaptation of this algorithm of Theorem \ref{tw_DP_THM}, that works for the maximization version of the problem instead of the counting version (albeit not optimally). We first describe the necessary modifications to the counting version and then the subsequent changes for use of our rounded values.

The algorithm for the maximization version needs the following changes: for a leaf node $i$ we set $D_i[s_0]\coloneqq1$, if $s_0=0$, and 0 otherwise. For an introduce node $i$, we also add a $+1$ to the values of previously computed entries if $s_{t+1}=0$ and the same conditions hold as in the counting version, while a value of 0 for invalid state-representations is substituted by an arbitrarily large negative value $-\infty$. For forget nodes $i$ we now compare all previous partial solutions to retain the maximum over all states of the forgotten vertex, instead of computing their sum, while for join nodes, we also substitute taking the sum by taking the maximum, with multiplication also substituted by addition of entries from the previous tables (i.e.\ we move our computations from the sum-product ring to the max-sum semiring), as well as subtracting from each such computation the number of vertices of zero state for the given entry (that would be counted twice).

We next explain the necessary modifications to the exact algorithm for use of the rounded states $\sigma\in\Sigma_\delta^d$. Consider a node $i$ introducing vertex $v_{t+1}$: for a new entry to describe a proper extension to some previously computed partial solution, if the new vertex is of state $s_{t+1}\in[1,d-1]$ in the new entry, then there must be some vertex $v_j\in X_i$, such that $s_{t+1}\le d(v_{t+1},v_j)+s_j$ (the one for which this sum is minimized), i.e.\ we require that the new state of the introduced vertex matches its distance to some other vertex in the bag plus the state of that vertex (being the one responsible for connecting $v_{t+1}$ to the partial solution). The rounded state $\sigma_{t+1}$ for $v_{t+1}$ must then satisfy: $\sigma_{t+1}\le d(v_{t+1},v_j)\oplus \sigma_j$.

Further, states are now considered low if $0<\sigma\le\frac{\lfloor d/2\rfloor}{(1+\epsilon)}$, while, from a set of already computed states $\sigma'$, the symmetrical (around $d/2$) state $\bar{\sigma}$ for a given low state $\sigma$ is defined as the minimum state $\sigma'$ for which $\sigma+\sigma'\ge\frac{d}{(1+\epsilon)}$. Thus, for a node introducing vertex $v_{t+1}$ with state $\sigma_{t+1}=0$, we require that $\forall v_j\in X_{i-1}$ with $\sigma_j=0$, it is $d(v_{t+1},v_j)\ge\frac{d}{(1+\epsilon)}$, and $\forall v_j\in X_{i-1}$ with $d(v_{t+1},v_j)\le\frac{\lfloor d/2\rfloor}{(1+\epsilon)}$, it is $\sigma_j\le d(v_{t+1},v_j)$ and $\sigma'=\bar{\sigma}$ for $D_i[\sigma_0,\dots,\sigma_{t+1}]\coloneqq D_{i-1}[\sigma'_0,\dots,\sigma_t']+1$. Finally, for join nodes, we arbitrarily choose the computed states for the table of one of the children nodes to represent the new entries and again use $\bar{\sigma}$ to identify the symmetrical of each low state (from the other node's table).


Moreover, we require that the tree decompositions on which our algorithm
is to be applied are rooted and of maximum depth $O(\log n)$. In \cite{BodlaenderH98} (Lemma 1), it is shown that
any tree decomposition of width $\tw$ can be converted to a rooted and binary
tree decomposition of depth $O(\log n)$ and width at most $3\tw+2$
in $O(\log n)$ time and $O(n)$ space. The following lemma employs the
transformation to bound the error of any value calculated in this way, based on
an appropriate choice of $\delta$ and therefore set $\Sigma_\delta^d$ of
available values, by relating the rounded states $\sigma$ computed at any node to the states $s$ that the exact algorithm would use at the same node instead.

\begin{lemma}\label{approx_bound_lem}
 Given $\epsilon$ and a tree decomposition $(\mathcal{X},T)$ with $T=(I,F),\mathcal{X}=\{X_i|i\in I\}$, where $T$ is rooted, binary and of depth $O(\log n)$, there exists a constant $C$, such that for all rounded states $\sigma_j\in\Sigma_\delta^d$ it is $\sigma_j\ge\frac{s_j}{(1+\epsilon)},\forall v_j\in X_i,\forall i\in I$, where $\delta=\frac{\epsilon}{C\log n}$.
\end{lemma}
\begin{proof}
 First, observe that for any rounded state $\sigma$ calculated using the $\oplus$ operator we have $\sigma\le s$, where $s$ is the state the exact algorithm would use instead. Let $h$ be the maximum depth of the recursive computations of any state $\sigma$ we may require. We now want to show by induction on $h$ that it is always $\log_{1+\delta}(\frac{s}{\sigma})\le h$. For $h=1$ and only one addition $\sigma=0\oplus d_1$, for some distance $d_1$ with $s=0+d_1$, we want $\log_{(1+\delta)}(\frac{s}{\sigma})\le1$. It is indeed $\log_{(1+\delta)}(\frac{s}{\sigma})=\log_{(1+\delta)}(d_1)-\lfloor\log_{(1+\delta)}(d_1)\rfloor\le1$.
 
 For the inductive step, let $\sigma_3=\sigma_2\oplus d_2$ and $s_3=s_2+d_2$ be the final rounded and exact values (at depth $h$), for some distance $d_2$ and previous values $\sigma_2,s_2$ (for $h-1$). It is $\log_{(1+\delta)}(\frac{s_3}{\sigma_3})=\log_{(1+\delta)}(\frac{s_2+d_2}{(1+\delta)^{\lfloor\log_{(1+\delta)}(\sigma_2+d_2)\rfloor}})=\log_{(1+\delta)}(s_2)+\log_{(1+\delta)}(1+\frac{d_2}{s_2})-\lfloor\log_{(1+\delta)}(\sigma_2)+\log_{(1+\delta)}(1+\frac{d_2}{\sigma_2})\rfloor$. This, after removal of the floor function, is $\le \log_{(1+\delta)}(s_2)+\log_{(1+\delta)}(1+\frac{d_2}{s_2})-(\log_{(1+\delta)}(\sigma_2)+\log_{(1+\delta)}(1+\frac{d_2}{\sigma_2}))+1=\log_{(1+\delta)}(s_2)-\log_{(1+\delta)}(\sigma_2)+\log_{(1+\delta)}(1+\frac{d_2}{s_2})-\log_{(1+\delta)}(1+\frac{d_2}{\sigma_2})+1$. The claim then follows, because $\log_{(1+\delta)}(s_2)-\log_{(1+\delta)}(\sigma_2)=\log_{(1+\delta)}(\frac{s_2}{\sigma_2})\le h-1$ by the inductive hypothesis, while also $\log_{(1+\delta)}(1+\frac{d_2}{s_2})-\log_{(1+\delta)}(1+\frac{d_2}{\sigma_2})\le0$, as $\sigma_2\le s_2$.
 
 Thus we have $\log_{(1+\delta)}(\frac{s}{\sigma})\le h$, from which we get $\frac{s}{\sigma}\le(1+\delta)^h$. For $\sigma\ge\frac{s}{(1+\epsilon)}$, we require that $(1+\delta)^h\le(1+\epsilon)$, or $h\le\log_{(1+\delta)}(1+\epsilon)=\frac{\log_2(1+\epsilon)}{\log_2(1+\delta)}$, that gives $h\le\frac{\epsilon}{\delta}$, for $\epsilon,\delta\approx0$, or $\delta\le\frac{\epsilon}{h}$. Next, observe that during the computations of the algorithm, the maximum depth $h$ of any computation can only increase by one if some vertex is introduced in the tree decomposition, as paths to and from it become available. This means no inductive computation we require can be of depth larger than the depth of the tree decomposition $T$, giving $h=C\log n$ for some constant $C$.  
\end{proof}

\begin{theorem}\label{tw_fptas}
 There is an algorithm which, given an edge-weighted instance of \dS\ $[G,k,d]$, a tree decomposition of $G$ of width $\tw$ and a parameter $\epsilon>0$, runs in time $O^*((\tw/\epsilon)^{O(\tw)})$ and finds a $d/(1+\epsilon)$-scattered set of size $k$, if a $d$-scattered set of the same size exists in $G$.
\end{theorem}
\begin{proof}
Naturally, our modified algorithm making use of these rounded values to represent the states will not perform the same computations as the exact version given in Section \ref{sec_tw_dp}. The new statement of correctness, taking into account the approximate values now computed (and the switch to the maximization version), is the following:

\begin{gather}
 \forall i\in I,\forall (\sigma_0,\dots,\sigma_t)\in(\Sigma_{\delta}^d)^{t+1},0\le t<\textrm{tw}:\label{tw_approx_cor_1}\\
 \big\{D_i[\sigma_0,\dots,\sigma_t]=|K|:K\subseteq V_i\setminus X_i\cup\{v_l\in X_i|\sigma_l=0\}: \label{tw_approx_cor_2}\\
 \{\forall u,w\in K:d(u,w)\ge \frac{d}{(1+\epsilon)}\}\wedge\label{tw_approx_cor_3}\\
 \wedge\{\forall v_j\in X_i|0<\sigma_j\le\frac{\lfloor d/2\rfloor}{(1+\epsilon)},\forall u,w\in K|d(u,v_j)\le d(w,v_j):\label{tw_approx_cor_4}\\
 d(u,v_j)\ge \sigma_j\wedge d(w,v_j)\ge \frac{d}{(1+\epsilon)}-\sigma_j \}\wedge \label{tw_approx_cor_5}\\
 \wedge \{\forall v_j\in X_i|\sigma_j>\frac{\lfloor d/2\rfloor}{(1+\epsilon)},\forall u\in K: d(u,v_j)\ge \sigma_j\}\label{tw_approx_cor_6}\vee\\
 \vee D_i[\sigma_0,\dots,\sigma_t]=-\infty\big\}.\label{tw_approx_cor_7}
\end{gather}

In words, the above states that for every node $i$ and all possible state-configurations $(\sigma_0,\dots,\sigma_t)\in(\Sigma_{\delta}^d)^{t+1}$ (\ref{tw_approx_cor_1}), table entry $D_i[\sigma_0,\dots,\sigma_t]$ contains the size of a subset $K$ of $V_i$ (that includes vertices $v_l\in X_i$ of state $s_l=0$) (\ref{tw_approx_cor_2}), such that the distance between every pair of vertices $u,w$ in $K$ is at least $d/(1+\epsilon)$ (\ref{tw_approx_cor_3}), for every vertex $v_j\in X_i$ of low state $\sigma_j\le\lfloor d/2\rfloor/(1+\epsilon)$ and a pair of vertices $u,w$ from $K$ with $u$ closer to $v_j$ than $w$ (\ref{tw_approx_cor_4}), its distance to $u$ is at least equal to its state $\sigma_j$ and its distance to $w$ is at least $d/(1+\epsilon)-\sigma_j$ (\ref{tw_approx_cor_5}), while for every vertex $v_j\in X_i$ of high state $\sigma_j>\lfloor d/2\rfloor/(1+\epsilon)$ and a vertex $u$ from $K$, its state $\sigma_j$ is at most its distance from any vertex $u$ (\ref{tw_approx_cor_6}), or if there is no such $K$, we have $D_i[\sigma_0,\dots,\sigma_t]=-\infty$ for this entry (\ref{tw_approx_cor_7}). This is shown by induction on the nodes $i\in I$:
\begin{itemize}
 \item Leaf node $i$ with $X_i=\{v_0\}$: This is the base case of our induction and the initializing values of 1 for $\sigma_0=0$ and 0 for $\sigma_0>0$ are indeed the correct sizes for $K$.
 \item Introduce node $i$ with $X_i=X_{i-1}\cup\{v_{t+1}\}$: For entries with $0<\sigma_j\le\lfloor d/2\rfloor/(1+\epsilon)$, validity of (\ref{tw_approx_cor_3},\ref{tw_approx_cor_6}) is not affected, while for (\ref{tw_approx_cor_4}-\ref{tw_approx_cor_5}): it is $\sigma_{t+1}=\sigma_j\oplus d(v_{t+1},v_j)$ for some vertex $v_j\in X_{i-1}$, for which, by the induction hypothesis we have that $\sigma_j\le d(u,v_j)$ and $d(w,v_j)\ge d/(1+\epsilon)-\sigma_j$, where $u$ is the closest selection to $v_j$ and $w$ the second closest. To see the same holds for $v_{t+1}$, observe that $\sigma_{t+1}\le d(v_{t+1},v_j)+d(u,v_j)=d(u,v_{t+1})$ and $d(w,v_j)\ge d/(1+\epsilon)-\sigma_j\Rightarrow \sigma_j+d(v_{t+1},v_j)\ge d/(1+\epsilon)-d(w,v_j)+d(v_{t+1},v_j)\Rightarrow \sigma_j+d(v_{t+1},v_j)\ge d/(1+\epsilon)-d(w,v_{t+1})\Rightarrow \sigma_{t+1}\ge d/(1+\epsilon)-d(w,v_{t+1})$.
 
 For entries with $\sigma_{t+1}>\lfloor d/2\rfloor/(1+\epsilon)$, validity of (\ref{tw_approx_cor_3}-\ref{tw_approx_cor_5}) is not affected, while for (\ref{tw_approx_cor_6}): it is $\sigma_{t+1}\le d(v_{t+1},v_j)+\sigma_j$ for some $v_j\in X_{i-1}$, for which we have $\sigma_j\le d(u,v_j)$ and thus also $\sigma_{t+1}\le d(v_{t+1},v_j)+d(v_j,u)=d(v_{t+1},u)$.
 
 For entries with $\sigma_{t+1}=0$, observe that the low states $\sigma_j$ of vertices $v_j\in X_{i-1}$ in the new entry with $d(v_{t+1},v_j)\le\lfloor d/2\rfloor/(1+\epsilon)$ would need to be $\sigma_j\le d(v_{t+1},v_j)$ and also correspond to the minimum high original state $\sigma_j'$ such that $\sigma_j+\sigma'_j\ge d/(1+\epsilon)$, for which partial solution it is $d(v_j,u)\ge \sigma_j',\forall u\in K$ and thus $d(v_j,u)\ge d/(1+\epsilon)-\sigma_j$ (\ref{tw_approx_cor_4}-\ref{tw_approx_cor_5}). For high states $\sigma_j$ of vertices $v_j\in X_{i-1}$, it is $d(v_{t+1},v_j)\ge \sigma_j$ (\ref{tw_approx_cor_6}) and finally, for (\ref{tw_approx_cor_3}), if there was some $u\in K$ such that $u\notin X_{i-}$ and $d(u,v_{t+1}<d/(1+\epsilon)$, then there must be some $v_j\in X_{i-1}$ of new state $\sigma_j$ and previous state $\sigma'_j$ (on the path between $u$ and $v_{t+1}$) for which $\sigma_j+\sigma'_j<d/(1+\epsilon)$, contradicting the requirement for introduction of $v_{t+1}$ with $\sigma_{t+1}=0$: it is $d(u,v_{t+1})<d/(1+\epsilon)\Rightarrow d(u,v_j)+d(v_{t+1},v_j)<d/(1+\epsilon)\Rightarrow \sigma_j+\sigma'_j<d/(1+\epsilon)$, as it must be $\sigma'_j\le d(u,v_j)$ and also $\sigma_j\le d(v_{t+1},v_j)$. 
 \item Forget node $i$ with $X_i=X_{i-1}\setminus\{v_{t+1}\}$: No modification to the exact dynamic programming affects the correctness of (\ref{tw_approx_cor_1}-\ref{tw_approx_cor_6}), as, the right number is indeed the maximum over all states for $v_{t+1}$.
 \item Join node $i$ with $X_i=X_{i-1}=X_{i-2}$: For (\ref{tw_approx_cor_3}), if there was a pair $u\in K\cap V_{i-1}\setminus X_{i-1}$, $w\in K\cap V_{i-2}\setminus X_{i-2}$ with $d(u,w)<d/(1+\epsilon)$, then there must be some vertex $v_j\in X_i$ (on the path between the two) for which $\sigma_j+\bar{\sigma_j}<d/(1+\epsilon)$ (as above). For (\ref{tw_approx_cor_4}-\ref{tw_approx_cor_6}), observe that for vertices of low state $\sigma_j$, lines (\ref{tw_approx_cor_4}-\ref{tw_approx_cor_5}) must have been true for either $i-1$ or $i-2$ and (\ref{tw_approx_cor_6}) for the other, while for vertices $v_j$ of high state $\sigma_j$, it again suffices that (\ref{tw_approx_cor_6}) must have been true for both.
\end{itemize}


For a node $i\in I$, let $U_i(k,s_0,\dots,s_t)=\{K\subseteq V_i|K\cap X_i= \{v_j\in\ X_i|s_j=0\}\}$ be the set of all $d$-scattered sets in $G_i$ of size $k$ for this state-configuration $(s_0,\dots,s_t)$ (as in the proof of Theorem \ref{tw_DP_THM}), $U_i(k,\sigma_0,\dots,\sigma_t)=\{K\subseteq V_i|K\cap X_i= \{v_j\in\ X_i|\sigma_j=0\}$ be the set of all subsets $K$ of $V_i$ of size $k$ for the rounded state-configuration $(\sigma_0,\dots,\sigma_t)$ (computed by our approximation algorithm) and $U_i(k,\frac{s}{1+\epsilon})$ be the set of all $d/(1+\epsilon)$-scattered sets of size $k$ in $G_i$. Consider a set $K\in U_i(k,s_0,\dots,s_t)$ and let $(\sigma_0,\dots,\sigma_t)$ be the state-configuration resulting from rounding each $s_j$ down to its closest integer power of $(1+\delta)$, or $\sigma_j=(1+\delta)^{\lfloor\log_{(1+\delta)}(s_j)\rfloor},\forall j\in[0,t]$. As $|K|=k$ and for any pair $u,w\in K$, we have $d(u,w)\ge d>\frac{d}{(1+\epsilon)}$, we want to show that the requirements of $(\sigma_0,\dots,\sigma_t)$ also hold for $K$.
By Lemma \ref{approx_bound_lem}, we know that $s_j\ge\sigma_j\ge\frac{s_j}{(1+\epsilon)}$ for all $j\in[0,t]$. Now, for each $\sigma_j\in (\sigma_0,\dots,\sigma_t)$, $\sigma_j\le s_j$ gives $\sigma_j\le d(u,v_j)$ for the closest $u\in K$ to $v_j$, while if also $\sigma_j\le\frac{\lfloor d/2\rfloor}{(1+\epsilon)}$, then $\frac{s_j}{(1+\epsilon)}\le\frac{\lfloor d/2\rfloor}{(1+\epsilon)}$ and $s_j\le\lfloor d/2\rfloor$ (i.e.\ $s_j$ is also low) and we have that $d-s_j\le d(w,v_j)$ for $w\in K$ being the second closest to $v_j$, from which we get $\frac{d}{(1+\epsilon)}-\frac{s_j}{(1+\epsilon)}\le \frac{d(w,v_j)}{(1+\epsilon)}\Rightarrow\frac{d}{(1+\epsilon)}-\sigma_j\le \frac{d(w,v_j)}{(1+\epsilon)}\le d(w,v_j)$, i.e.\ state-configuration $(\sigma_0,\dots,\sigma_t)$ also holds for set $K$. This means $K\in U_i(k,\sigma_0,\dots,\sigma_t)$ and thus we have $U_i(k,s_0,\dots,s_t)\subseteq U_i(k,\sigma_0,\dots,\sigma_t)$. Further, since for any $K$ our approximation algorithm will compute, it is $d(u,w)\ge d/(1+\epsilon), \forall u,w\in K$, we also have $U_i(\sigma_0,\dots,\sigma_t)\subseteq U_i(k,\frac{s}{(1+\epsilon)})$. Due to these considerations, if a $d$-scattered set of size $k$ exists in $G$, our algorithm will be able to return a set $K$ with $|K|=k$, that will be a $d/(1+\epsilon)$-scattered set of $G$. 

The algorithm is then the following: first, according to the statement of Lemma \ref{approx_bound_lem}, we select $\delta=\frac{\epsilon}{C\log n}$, that we use to define set $\Sigma_{\delta}^d$ and then we use the algorithm of Theorem \ref{tw_DP_THM}, modified as described above, on the bounded-height transformation of nice tree decomposition $(\mathcal{X},T)$. Correctness of the algorithm and justification of the approximation bound are given above, while the running time crucially depends on the size of $\Sigma_{\delta}^d$ being $|\Sigma_{\delta}^d|=O(\log_{(1+\delta)}d)=O(\frac{\log d}{\log(1+\delta)})=O(\frac{\log d}{\delta})$, where we used the approximation $\log(1+\delta)\approx\delta$ for sufficiently small $\delta$ (i.e.\ sufficiently large $n$). This gives $O(\log n/\epsilon)^{O(\tw)}$ and the statement is then implied by Lemma \ref{lem:fpt-logn}.
 
As a final note, observe that due to the use of the $\lfloor\cdot\rfloor$ function in the definition of our $\oplus$ operator, all our values will be rounded \emph{down}, in contrast to the original version of the technique (from \cite{Lampis14}), where depending on a randomly chosen number $\rho$, the values could be rounded either down or up. This means there will be some value $x$, such that $x\oplus1=x$, or $(1+\delta)^x=(1+\delta)^{\lfloor\log_{(1+\delta)}(x+1)\rfloor}$ (we would have $x\approx1/\delta$). One may be tempted to conceive of a pathological instance consisting of a long path on $n$ vertices and $d>>x$, along with a simple path decomposition for it (that is essentially of the same structure), where the computations for each rounded state $\sigma$ would ``get stuck'' at this value $x$. In fact, the transformation of \cite{BodlaenderH98} would give a tree decomposition of height $O(\log n)$ for this instance, whose structure would be the following: the leaf nodes would correspond to one vertex of the path each, while at (roughly) each height level $i$, sub-paths of length $2^i$ would be joined together. Thus each join node $t$ that corresponds to some sub-path of length $2^i$ (let $X_t=\{a,b,c,d\}$) would have two child branches, consisting of two forget nodes, two introduce nodes and a previous join node on each side (let these be $t-1,t-2$), computing sub-paths of length $2^{i-1}$ (with $X_{t-1}=\{a,a',b',b\}$ and $X_{t-2}=\{c,c',d',d\}$). The vertices forgotten at each branch would be the middle vertices of the sub-path of length $2^{i-1}$ already computed at the previous join node of this branch (i.e.\ $a',b'$ for the $t-1$ side and $c',d'$ for the other), while the introduced vertices would be the endpoints of the sub-path of length $2^{i-1}$ computed at the other branch attached to this join node (i.e.\ $c,d$ for the $t-1$ side and $a,b$ for the other). In this way, in each branch (and partial solution) there will be one vertex ($c$ or $b$) for which the rounded state would need to be $\le1\oplus$ the rounded state $\sigma$ of some neighbor ($b$ or $c$) and one vertex ($d$ or $a$) for which the $\oplus$ operator would be applied between the state $\sigma$ of some non-adjacent vertex ($a,b$ or $c$ for the $t-1$ side and $c,d$ or $b$ for the other) and their distance (e.g.\ $d(b,d)$ and $d(c,a)$), these being at least $2^{i-1}$. In this way, the algorithm will not have to compute any series of rounded states sequentially by $\oplus1$ and as, by Lemma \ref{approx_bound_lem}, we have that for all nodes $i\in I$ and vertices $v_j\in X_i$, it is $\sigma_j\ge\frac{s_j}{(1+\epsilon)}$, for all $\sigma_j\in\Sigma_{\delta}^d$, the rounded states used by the algorithm for these introduce/join nodes will never be more than a factor of $(1+\epsilon)$ from the ones used by the exact algorithm on the same tree decomposition.  
\end{proof}

\section{Conclusion}
In this paper we considered the \dS\ problem, a distance-based generalization of \textsc{Independent Set}. We focus on structural parameterization, due to the problem's well-investigated hardness and inapproximability. In particular, we give tight fine-grained bounds on the complexity of \dS\ with respect to the well-known graph parameters treewidth $\tw$, tree-depth $\td$, vertex cover $\vc$ and feedback vertex set $\fvs$:
\begin{itemize}
 \item A Dynamic Programming algorithm of running time $O^*(d^{\tw})$ and a matching lower bound based on the SETH, that generalize known results for \textsc{Independent Set}.
 \item W[1]-hardness for parameterization by $\vc+k$ for edge-weighted graphs, as well as by $\fvs+k$ for unweighted graphs, while these are complemented by an FPT-time algorithm for $\vc$ and the unweighted case.
 \item An algorithm solving the problem for unweighted graphs in time $O^*(2^{O(\td)^2})$ and a matching ETH-based lower bound.
 \item An algorithm computing for any $\epsilon>0$ a $d/(1+\epsilon)$-scattered set in time $O^*((\tw/\epsilon)^{O(\tw)})$, if a $d$-scattered set exists in the graph, assuming a tree decomposition of width $\tw$ is provided along with the input.
\end{itemize}
Remaining open questions on the structurally parameterized complexity of the problem concern the identification of similarly tight upper and lower (SETH-based) bounds for \dS\ parameterized by the related parameter clique-width, as well as the sharpening of our ETH-based lower bounds for $\vc$ and $\fvs$, that are not believed to be tight due to the quadratic blow-up in parameter size in our reductions.

\bibliography{d_S_citations}

\end{document}